\documentclass[12pt,twoside]{amsart} 
\usepackage[papersize={210mm,297mm},left=15mm,right=15mm,top=20mm,bottom=20mm]{geometry}
\usepackage{graphicx}
\usepackage{mathrsfs}
\usepackage{color}
\usepackage{amssymb}
\usepackage{amsthm}
\usepackage{xcolor}
\definecolor{slategray}{RGB}{112,138,144}

\usepackage[colorlinks=true]{hyperref}
\hypersetup{linkcolor=slategray,urlcolor=slategray,citecolor=slategray}
\usepackage[preserveurlmacro]{breakurl}

\newcommand{\E}{{\mathbb{E}}}
\newcommand{\I}{{\rm i}}
\newcommand{\dd}{{\rm d}}

\newcommand{\beq}{\begin{equation}}
\newcommand{\eeq}{\end{equation}}
\newcommand{\neeq}{\nonumber\end{equation}}

\newcommand{\sfrac}[2]{{\textstyle \frac{#1}{#2}}}

\newcommand{\del}{\partial}

\newcommand{\cH}{{\mathcal H}}
\newcommand{\bC}{{\mathbb C}}
\newcommand{\bN}{{\mathbb N}}
\newcommand{\bR}{{\mathbb R}}
\newcommand{\bX}{{\mathbb X}}
\newcommand{\bZ}{{\mathbb Z}}

\newcommand{\veps}{\varepsilon}

\newcommand{\vphi}{\varphi}

\newcommand{\ps}{\psi}

\newcommand{\psq}{{\bar\psi}}

\newcommand{\cP}{{\mathcal{P}}}

\usepackage{graphicx}
\usepackage{appendix}
\usepackage{cases}

\newtheorem{theo}{\textsc{Theorem}}[section]
\newtheorem{lemma}{\textsc{Lemma}}[section]
\newtheorem{defi}{\textsc{Definition}}[section]

\newtheorem{prop}{\textsc{Proposition}} [section]
\newtheorem{rem}{\textsc{Remark}}[section]

\numberwithin{equation}{section}

\DeclareMathOperator{\xx}{{\bf{x}}}

\newcommand{\reftitel}[1]{{\sl #1}}
\newcommand{\enquote}[1]{``#1''}

\newcommand{\alatt}{\veps}
\newcommand{\mass}{{\rm m}}

\newcommand{\BQ}{Q}

\newcommand{\Ftn}[1]{\tilde F^{(#1)}}  
\newcommand{\Fphz}{F^{(0)}}
\newcommand{\phz}{\phi^{(0)}}

\begin{document}

\title[]{The majorant method for the fermionic effective action}
\author{Wilhelm Kroschinsky, Domingos H. U. Marchetti, Manfred Salmhofer} 
\address{
Institute for Applied Mathematics, Hausdorff-Center for Mathematics, University of Bonn, Endenicher Allee 60, 53115 Bonn, Germany
 --- 
Instituto de F\' isica, Universidade de São Paulo, Brazil
 --- 
Institut f\" ur Theoretische Physik, Universit\" at Heidelberg,
Philosophenweg~19, 69120 Heidelberg, Germany}

\email{kroschinsky@iam.uni-bonn.de,marchett@if.usp.br,salmhofer@uni-heidelberg.de}

\date{\today}

\begin{abstract}
\noindent
We revisit the problem of controlling Polchinski's equation by the solution of an associate Hamilton-Jacobi equation which determines a norm majorant for the fermionic effective action. This method, referred to as the majorant method, was first introduced by D. Brydges and J. Wright in 1988, but its original formulation contains a gap which has never been addressed. We overcome this gap and show that the majorant equation and its existence condition are analogous to the ones originally obtained by Brydges and Wright. As an application of the method, we investigate a fermion model with a local quartic interaction. 
\end{abstract}

\maketitle

\section{Introduction}

Polchinski's equation \cite{Polch} has revolutionized proofs of perturbative renormalizability, by shifting the emphasis in proofs from individual graphs to the Green functions of the quantum field theory. The structure of Polchinski's equation is that of a nonlinear heat equation in field space, which is infinite-dimensional in a continuum theory and very high-dimensional in most regularizations, such as the lattice regularization. 

Starting with \cite{KKS}, his method has been developed to the point where almost all results in perturbative renormalization have been reproven \cite{KeKoQED,KeKoZi,KoMue}, and many have been extended significantly \cite{HoHo,KoHoOPE,HoStM,Sa98,MSbook}. 
The simplicity of the strategy of proof makes one want to use it also nonperturbatively, i.e.\ at small coupling, but not in the framework of formal power series. One case where this has been done is in the field-theoretical formulation of the KAM problem \cite{DeSiKupi}, which is, however, the case where the loops of quantum field theory are absent. 

Pioneering work relating Polchinski's equation and well-established methods of constructive quantum field theory was first done by Brydges and Kennedy. They showed that tree and polymer expansions can be obtained as explicit (convergent) series solutions of the differential equation, and that one can give a majorant for the solution of Polchinski's equation that satisfies a Hamilton-Jacobi equation. In a further paper, Brydges and Wright studied this majorant for fermionic systems \cite{BryWri}. It was later discovered that the proof in \cite{BryWri} contains a gap \cite{BryWriErr}, and it has since remained an open problem to complete the proof. This is the first purpose of the present paper. 

In the mean time, much work was devoted to simplifying and extending the fermionic tree expansions. While many works rely on the Brydges-Battle-Federbush interpolation, one can also use Polchinski's equation to generate an interpolation formula that has the necessary positivity properties. Continuous scale decompositions and flows were also used to construct the Gross-Neveu model in \cite{DR-GN} and two-dimensional many-fermion systems in \cite{DR1,DR2}. 

These constructions are, however, not as simple as the original Polchinski proof -- partly this is an inherent difference between perturbative and non-perturbative methods, but the strategy appears closer to that of tree expansions. The second purpose of this paper is to indicate how Polchinski-type estimates may be done using the majorant method. This is only an exemplification corresponding to the restriction to `completely convergent graphs' in other approaches; a full construction of fermionic models is possible along these lines, but is deferred to future work. 

Let us say a few words about how this paper is organized. For the reader's convenience, we start by recalling some standard definitions and results from the theory of Grassmann algebra. This is done in Section \ref{grassmannoverview}. In Section \ref{norms}, we quickly introduce the relevant norms to be used on the Grassmann algebra, which are in accordance with Brydges and Wright's paper. Section \ref{rgtransformation} is devoted to introducing a the renormalization group transformation and its connection to Polchinsky's equation. Our main result, Theorem \ref{newmajorant}, is stated and proved in Section \ref{majorantmethod}, where the Majorant Method is developed. Finally, some applications of the method using quartic interactions as initial condition are then discussed in Section \ref{quarticapplications}. 

Note added: After completion of this paper we learned about the related works of Pawel Duch \cite{Duch}, who constructs the two-dimensional Gross-Neveu model using the Polchinski equation, and Rafael Greenblatt \cite{Greenblatt}.

\section{Grassmann Algebras - An Overview}\label{grassmannoverview}

Fermionic systems are naturally described in terms of Grassmann algebras \cite{berezin2012method}. For this reason, it is convenient to reserve this first section to introduce some notations as well as some standard, yet relevant, concepts regarding these algebras.

Let $V$ and $Z$ be vector spaces over $\mathbb{K} = \mathbb{R}$ or $\mathbb{C}$ and $p \in \mathbb{N}$ with $p \ge 2$. We recall that a $p$-linear map $T: \overbrace{V\times \cdots \times V}^{\text{$p$ times}} \to Z$ is called $p$-alternating if
\begin{eqnarray}
T(\varphi_{1},...,\varphi_{i},...,\varphi_{k},...,\varphi_{p}) = -T(\varphi_{1},...,\varphi_{k},...,\varphi_{i},...,\varphi_{p}) \label{GA1}
\end{eqnarray}
holds for every $\varphi_{1},...,\varphi_{i},...,\varphi_{k},...,\varphi_{p}\in V$. 

\begin{defi}\label{exteriorpower} A $p$-alternating map $\wedge^{p}: V \times \cdots \times V \to Z$ is said to satisfy the universal property for alternating maps if given another vector space $W$ and a $p$-alternating map $T: V \times \cdots \times V \to W$, there exists a unique linear map $\wp : Z \to W$ such that $T = \wp \circ \wedge^{p}$.

Under these conditions, the ordered pair $(Z,\wedge^{p})$ is called a $p$-th exterior power of $V$.
\end{defi}

If $(Z,\wedge^{p})$ is a $p$-th exterior power of $V$, it is customary to write $Z = \bigwedge^{p}V$. This already embodies the $\wedge^{p}$ dependence, so we refer to $\bigwedge^{p}V$ itself as the $p$-th exterior power of $V$. 

We remark that the universal property implies that decomposable vectors of the form $\wedge^{p}(\varphi_{1},...,\varphi_{p}) \equiv \varphi_{1}\wedge \cdots \wedge \varphi_{p}$ generate $\bigwedge^{p} V$ \cite{Greub}, that is, $\operatorname{Span}\wedge^{p} = \bigwedge^{p}V$. Notice that, because $\bigwedge^{p}$ is $p$-alternating, such decomposable vectors are identically zero whenever $\varphi_{i} = \varphi_{j}$ for at least one pair of indices $i \neq j$. In particular, $\varphi \wedge \varphi \equiv \varphi^{2} = 0$ holds true for every $\varphi \in V$. 

The $p$-th exterior power of $V$ can be concretely realized as the quotient space of $\bigotimes^{p}V := \overbrace{V\otimes \cdots \otimes V}^{\text{$p$ times}}$ with the space of all decomposable tensors $\varphi_{1}\otimes \cdots \otimes \varphi_{p}$ with at least one pair of coinciding entries. Moreover, $\bigwedge^{p}V$ can be naturally identified with the space of $p$-fold alternating tensors by means of the universal property. 

When $V$ is finite-dimensional with basis $\mathcal{V} = \{\psi_{1},...,\psi_{n}\}$, the set $\{\psi_{i_{1}}\wedge \cdots \wedge \psi_{i_{p}}: \hspace{0.1cm} 1 \le i_{1} < \cdots < i_{p}\le n\}$ forms a basis for $\bigwedge^{p}V$. Consequently, $\bigwedge^{p}V = \{0\}$ whenever $p > n$. 

Using the conventions $\bigwedge^{(0)} V := \mathbb{K}$ and $\bigwedge^{1}V := V$, for every fixed $p, q \in \mathbb{N}$ there exists a bilinear map $\wedge: \bigwedge^{p}V \times \bigwedge^{q}V \to \bigwedge^{p+q}V$ which associates, to each pair $(\varphi_{1},...,\varphi_{p},\phi_{1},...,\phi_{q})$, the element $\varphi_{1}\wedge \cdots \wedge \varphi_{p}\wedge \phi_{1}\wedge \cdots \wedge \phi_{q}$. Hence, the algebraic direct sum
\begin{eqnarray}
\bigwedge V := \bigoplus_{p=0}^{\infty}\bigwedge^{p}V, \label{GA2}
\end{eqnarray}
which is a vector space of dimension $2^{n}$, becomes an algebra over $\mathbb{K}$ when equipped with the product
\begin{eqnarray}
f\wedge g := \sum_{p=0}^{\infty}\bigg{(}\sum_{i+j = p}f_{i}\wedge g_{j}\bigg{)}, \label{GA3}
\end{eqnarray}
with $f = (f_{0},f_{1},...),g = (g_{0},g_{1},...) \in \bigwedge V$. This algebra is referred to as the \textit{Grassmann or Exterior algebra} over $\mathbb{K}$. Note that, by identifying each $(0,0,...,f_{p},0,...)$ with $f_{p}\in \bigwedge^{p}V$, any element $f \in \bigwedge V$ can be written as a polynomial of the form
\begin{eqnarray}
f = f(\Psi) = \sum_{p=0}^{\infty}\sum_{i_{1}<\cdots < i_{p}}\zeta_{i_{1},...,i_{p}}(\psi_{i_{1}}\wedge \cdots \wedge \psi_{i_{p}}), \label{GA4}
\end{eqnarray}
with coefficients $\zeta_{i_{1},...,i_{p}}\in \mathbb{K}$. The term $p=0$ in the sum is understood as a scalar. In short, $\bigwedge V$ is the algebra generated by $1$ and the \textit{fields} $\psi_{1},...,\psi_{n}$, which satisfy anti-commutation relations
\begin{eqnarray}
\psi_{i}\wedge \psi_{j} + \psi_{j}\wedge \psi_{i} = 0, \label{GA5}
\end{eqnarray}
for all $i,j=1,...,n$. For this reason, $1, \psi_{1},...,\psi_{n}$ are usually called the \textit{generators} of the Grassmann algebra. 

We call $f$ an \textit{even} (respectively \textit{odd}) element of the algebra if $\zeta_{i_{1},...,i_{p}} = 0$ for every odd (respectively even) $p$ in expression (\ref{GA4}). Of course, every element of the Grassmann algebra can be uniquely decomposed as a sum of an even and an odd element, since
\begin{eqnarray}
\bigwedge V \cong \bigoplus_{\text{$p$ even}}\bigwedge^{p}V \oplus \bigoplus_{\text{$p$ odd}}\bigwedge^{p}V. \nonumber
\end{eqnarray}

Sometimes we need to take other fields into account. Suppose that $U$ is another vector space over $\mathbb{K}$ with basis $\mathcal{U} = \{\theta_{1},...,\theta_{m}\}$, and let us consider the direct sum $V\oplus U$, which has basis $\{(\psi_{1},0),...,(\psi_{n},0),(0,\theta_{1}),...,(0,\theta_{m})\}$. The associated Grassmann algebra $\bigwedge (V\oplus U)$ is generated by $1$ and its basis elements. Under the identifications $(\psi_{i},0) \mapsto \psi_{i}$ and $(0,\theta_{j}) \mapsto \theta_{j}$, a $2$-form $(\psi_{i},0)\wedge (0,\theta_{j})$ becomes simply $\psi_{i}\wedge \theta_{j}$. As a consequence, any element $f \in \bigwedge (V\oplus U)$ can be written as a polynomial
\begin{eqnarray}
f = f(\Psi,\Theta) = \sum_{k+l \le n+m}\sum_{i_{1}<\cdots i_{p}}\sum_{j_{1}<\cdots < j_{l}}\zeta_{i_{1},...,i_{p},j_{1},...,j_{l}}(\psi_{i_{1}}\wedge \cdots \wedge \psi_{i_{p}}\wedge \theta_{j_{1}}\wedge \cdots \wedge \theta_{j_{l}}). \label{GA6}
\end{eqnarray}
In (\ref{GA6}), it is understood that $f$ has only $\theta$-variables when $k=0$ and only $\psi$-variables when $l=0$. Therefore, this expression provides a natural way of identifying both $\bigwedge V$ and $\bigwedge U$ as vector subspaces of $\bigwedge (V\oplus U)$. 

\begin{defi}\label{homogeneouselement} An element $f \in \bigwedge (V\oplus U)$ which depends only on the fields $\psi_{1},...,\psi_{n}$ (resp. $\theta_{1},...,\theta_{n}$) is called $\psi$-homogeneous (resp. $\theta$-homogeneous).
\end{defi}

When $n=m$, we can produce nonhomogeneous elements of $\bigwedge (V\oplus U)$ out of homogeneous ones by a simple translation of variables; if $f = f(\Psi)$ is $\psi$-homogeneous, then
\begin{eqnarray}
f(\Psi+\Theta) := f(\psi_{1}+\theta_{1},...,\psi_{n}+\theta_{n}) \label{GA7}
\end{eqnarray}
is nonhomogeneous. 

Let $f = f_{0}+ f_{1} \in \bigwedge V$, with $f_{0} \in \mathbb{C}$ and $f_{1} \in \bigoplus_{p=1}^{\infty}\bigwedge^{p}V$. Let $U \subseteq \mathbb{C}$ be an open set and $F: U \to \mathbb{C}$ be such that $f_{0} \in U$ and all derivatives $F^{(k)}(f_{0})$ of $F$ exist at $f_{0}$, for $k \le n$. Then, we can define $F(f) \equiv F(f(\Psi)) \in \bigwedge V$ by its Taylor series
\begin{eqnarray}
F(f) := F(f_{0}) + \sum_{k=1}^{\infty}\frac{F^{(k)}(f_{0})}{k!}\;{f_1}^{k}. \label{GA8}
\end{eqnarray}
In particular, we use this result to define the exponential $e^{f}$ for all $f \in \bigwedge V$ and the logarithm $\ln f$ for all those $f \in \bigwedge V$ that satisfy $f_0 > 0$. 

\subsection{Derivatives and Integrals on Grassman Algebras}\label{derivativesandintegrals}

To introduce derivatives on Grassmann algebras, we mimic the rules of partial derivatives of multivariable calculus. Suppose $\bigwedge V$ is generated by $1$ and the fields $\psi_{1},...,\psi_{n}$ and let $k \in \{1,...,n\}$ be fixed. We define the (partial) derivative of a monomial $\psi_{i_{1}}\wedge \cdots \wedge \psi_{i_{p}}$, $0 \le p \le n$, with respect to $\psi_{k}$ to be zero if $p=0$ and 
\beq
\frac{\partial}{\partial \psi_{k}}(\psi_{i_{1}}\wedge \cdots \wedge \psi_{i_{p}}) = \delta_{i_{1},k}\psi_{i_{2}}\wedge \psi_{i_{3}}\wedge  \cdots \wedge \psi_{i_{p}} -\delta_{i_{2},k}\psi_{i_{1}}\wedge \psi_{i_{3}}\wedge \cdots \wedge \psi_{i_{p}}+\cdots +(-1)^{p}\delta_{i_{p},k}\psi_{i_{1}}\wedge \psi_{i_{2}}\wedge \cdots \wedge \psi_{i_{p-1}} \label{GA11}
\eeq
otherwise, with $\delta_{n,m}$ denoting a Kronecker delta. We then extend the definition of derivative to $\bigwedge^{p}V$ by linearity.

\begin{defi}\label{derivativeswedge} If $f$ is given by (\ref{GA4}), then
\begin{eqnarray}
\frac{\partial f}{\partial \psi_{k}} := \sum_{p=0}^{\infty}\sum_{i_{1}<\cdots < i_{p}}\zeta_{i_{1},...,i_{p}}\frac{\partial}{\partial \psi_{k}}(\psi_{i_{1}}\wedge \cdots \wedge \psi_{i_{p}}). \label{GA13}
\end{eqnarray}
\end{defi}

Higher order derivatives are defined by iterative applications of partial derivatives
\begin{eqnarray}
\frac{\partial^{k}f}{\partial{\psi_{i_{1}}\cdots \partial \psi_{i_{k}}}} := \frac{\partial}{\partial \psi_{i_{1}}}\cdots \frac{\partial f}{\partial \psi_{i_{k}}}. \label{GA14}
\end{eqnarray}
In the context of Grassmann algebras, integrals and derivatives are defined in the same way.

\begin{defi}\label{integrals} The integral of $f \in \bigwedge V$ with respect to the fields $\psi_{j_{1}},...,\psi_{j_{k}}$ is 
\begin{eqnarray}
\int f d\psi_{j_{k}}\wedge \cdots \wedge d\psi_{j_{1}} := \frac{\partial}{\partial \psi_{j_{1}}}\cdots \frac{\partial f}{\partial \psi_{j_{k}}}. \label{GA15}
\end{eqnarray}
Consequently, the integral of a scalar is always zero whilst
\begin{eqnarray}
\int \psi_{m} d\psi_{n} = \delta_{m,n} \quad \mbox{and} \quad \int f d\psi_{1}\wedge \cdots \wedge d\psi_{n} = \zeta_{1,...,n}. \label{GA16}
\end{eqnarray}
\end{defi}

Before we move on, suppose $f = f_{p} \in \bigwedge^{p}V$ is given by
\begin{eqnarray}
f_{p}(x) =  \sum_{i_{1}<\cdots < i_{p}}\zeta_{i_{1},...,i_{p}}(x) (\psi_{i_{1}}\wedge \cdots \wedge \psi_{i_{p}}), \label{GA17}
\end{eqnarray}
where the coefficients $\zeta_{i_{1},...,i_{p}}$ are now differentiable functions of a real-variable $x$. In this case, we say that $f_{p}(x)$ is differentiable with respect to $x$ and define its derivative by 
\begin{eqnarray}
\frac{d f_{p}(x)}{dx} := \sum_{i_{1}<\cdots < i_{p}}\frac{d\zeta_{i_{1},...,i_{p}}(x)}{dx} (\psi_{i_{1}}\wedge \cdots \wedge \psi_{i_{p}}). \label{GA18}
\end{eqnarray}
If $f(x) = \sum_{p=0}^{\infty}f_{p}(x) \in \bigwedge V$, we define its derivative with respect to $x$ by
\begin{eqnarray}
\frac{d f(x)}{dx} := \sum_{p=0}^{\infty}\sum_{i_{1}<\cdots < i_{p}}\frac{d\zeta_{i_{1},...,i_{p}}(x)}{dx} (\psi_{i_{1}}\wedge \cdots \wedge \psi_{i_{p}}). \label{GA19}
\end{eqnarray}

Of course, many standard properties such as linearity and the chain rule for derivatives, as well as rules such as integration  by parts for integrals have natural counterparts in the context of Grassmann algebras, when these operations are defined as before. A much more detailed discussion on these topics can be found, e.g. in \cite{berezin2013introduction, MSbook, FKTbook}.

\subsection{Grassmann Gaussian Integrals}\label{gaussianintegrals}

From now on, let us fix $\mathbb{K} = \mathbb{C}$ and denote by $M_{m}(\mathbb{C})$ the space of all $m \times m$ matrices with complex entries. The determinant of a skew-symmetric matrix $A = (A_{ij}) \in M_{m}(\mathbb{C})$ must satisfy the condition $\det A = (-1)^{m}\det A$ and, for this reason, $A$ is singular when $m$ is odd. On the other hand, when $m=2n$ is even, $n \ge 1$, its determinant is the square of a polynomial $\operatorname{Pf}(A)$, called the Pfaffian of $A$, and given by
\begin{eqnarray}
\operatorname{Pf}(A) := \frac{1}{2^{n}(n)!}\;\sum_{\sigma \in S_{2n}}\operatorname{sgn}(\sigma)A_{\sigma(1)\sigma(2)}\cdots A_{\sigma(2n-1)\sigma(2n)}, \label{GA21}
\end{eqnarray}
where the sum ranges over all permutations $\sigma$ of the set $\{1,...,2n\}$. 

\begin{defi}
Let $V$ be a finite-dimensional vector space with basis $\{\psi_{1},...,\psi_{2n}\}$, $n \ge 1$, and $A = (A_{ij}) \in M_{2n}(\mathbb{C})$ skew-symmetric. The Grassmann Gaussian integral on $\bigwedge V$ with covariance (or propagator) $A$ is the linear map
\begin{eqnarray}
\bigwedge V \ni  f \mapsto \E_{\mu_{A}}[f] :=\int f(\Psi)d\mu_{A}(\Psi) \in \mathbb{C} \label{GA22}
\end{eqnarray}
determined by its correlation functions
\begin{eqnarray}
 \int \psi_{i_{1}}\wedge \cdots \wedge \psi_{i_{p}}d\mu_{A}(\Psi) := \operatorname{Pf}[A_{i_{k},i_{l}}]_{1 \le k,l \le p}. \label{GA23}
\end{eqnarray}
Notice that the pfaffian appearing in the right hand side of (\ref{GA23}) is zero when $p$ is an odd number. 
\end{defi}

We are usually interested in the case where $A$ is not only skew-symmetric, but also invertible. When this is the case, one has
\begin{eqnarray}
\E_{\mu_{A}}[f] = \operatorname{Pf}(A)\int f(\Psi)e^{-\frac{1}{2}\langle \Psi, A^{-1}\Psi\rangle} d\psi_{1}\wedge \cdots \wedge d\psi_{2n}. \label{GA24}
\end{eqnarray}
for every $f \in \bigwedge V$. In (\ref{GA24}), we are expressing fields as vectors
\begin{eqnarray}
\Psi = \begin{pmatrix}
\psi_{1} \\
\vdots \\
\psi_{2n}
\end{pmatrix}
\label{GA25}
\end{eqnarray}
and we have introduced an inner product notation
\begin{eqnarray}
\langle \Psi, \Theta \rangle := \Psi^{T}\Theta = \psi_{1}\wedge \theta_{1} + \cdots + \psi_{2n}\wedge \theta_{2n} \label{GA26}
\end{eqnarray}
which is very convenient when dealing with quadratic forms. In particular, a simple completion of squares argument leads to the following formula
\begin{eqnarray}
\E_{\mu_{A}}[e^{\langle \Psi,\Theta\rangle}] = e^{-\frac{1}{2}\langle \Theta, A\Theta\rangle}. \label{GA27}
\end{eqnarray}
Again, the right hand side truncates to a polynomial on the finite-dimensional Grassmann algebra generated by the $\Theta$'s.

\begin{prop} \label{propertiesgaussianmeasure}

\

\textup{(a)} If the weighted Laplacian operator $\Delta_{A}$ is defined by
\begin{equation}
\Delta_{A} := \bigg{\langle}\frac{\partial}{\partial \Psi}, A\frac{\partial}{\partial \Psi}\bigg{\rangle} \equiv \sum_{i,j=1}^{2n}A_{ij}\frac{\partial}{\partial \psi_{i}}\frac{\partial}{\partial \psi_{j}}  \label{GA28}
\end{equation}
then
\begin{equation}
(\mu_{A}*f)(\Psi) := \int f(\Psi+\Theta)d\mu_{A}(\Theta) = e^{\frac{1}{2}\Delta_{A}}f(\Psi)  \label{GA28.1}
\end{equation}
holds.\footnote{ The exponential of $\Delta_{A}$ truncates to a polynomial by nilpotency of the derivatives, so the right hand side of (\ref{GA28.1}) is well-defined and analytic in $A$.} We refer to (\ref{GA28.1}) as the Heat Kernel Formula. 

\textup{(b)}  If $A, B \in M_{2n}(\mathbb{C})$ are both skew-symmetric and invertible, then
\begin{equation}
\E_{\mu_{A+B}}[f] = \E_{\mu_{B}}[\mu_{A}*f] 
\quad
\mbox{and}
\quad
\mu_{A+B}*f = \mu_B * (\mu_A *f) \; .
\label{GA28.2}
\end{equation}
Relation (\ref{GA28.2}) is also known as the covariance splitting formula. 

\end{prop}
\begin{proof}
Part \textup{(a)} is proved in \cite{FKTbook,MSbook}. Part \textup{(b)} follows from (\ref{GA28.1}), the relations $\Delta_{A+B} = \Delta_A + \Delta_B$ and $\Delta_A\Delta_B=\Delta_B\Delta_A$.
\end{proof}

In many applications, the covariance of a fermionic theory can be expressed as a block matrix
\begin{eqnarray}
A = \begin{pmatrix}
0 & C \\
-C & 0 
\end{pmatrix}
\label{GA29.1}
\end{eqnarray}
for some symmetric matrix $C = (C_{ij}) \in M_{n}(\mathbb{C})$, in which case $\operatorname{Pf}(A) = \operatorname{det}C$. This motivates the split of our fields $\{\psi_{1},...,\psi_{2n}\}$ into two sets $\{\psi_{1},...,\psi_{n}\}$ and $\{\bar{\psi}_{1},...,\bar{\psi}_{n}\}$, with $\bar{\psi}_{i} := \psi_{i+n}$. Such transformations yield
\begin{eqnarray}
\frac{1}{2}\langle \Psi, A\Psi \rangle = \langle \bar{\Psi},C\Psi\rangle. \label{GA30}
\end{eqnarray}
When $C$ is invertible, we have
\begin{eqnarray}
A^{-1} = \begin{pmatrix}
0 & -C^{-1} \\
C^{-1} & 0 
\end{pmatrix}
.\label{GA31}
\end{eqnarray}
Using (\ref{GA24}) and (\ref{GA30}), it is natural to define
\begin{eqnarray}
\int f(\bar{\Psi},\Psi)d\mu_{C}(\bar{\Psi},\Psi) := \det C\int f(\bar{\Psi},\Psi)e^{-\langle \bar{\Psi},C^{-1}\Psi\rangle}d\bar{\psi}_{1}\wedge \cdots \wedge d\bar{\psi}_{n}\wedge d\psi_{1}\wedge \cdots \wedge d\psi_{n}, \label{GA32}
\end{eqnarray}
which is referred to as the Grassmann Gaussian integral with covariance $C$. 

\begin{prop}\label{gaussiancorrelationsbarredunbarred} Let $J = \{j_{1},...,j_{p}\}$ and $K = \{k_{1},...,k_{q}\}$ be nonempty subsets of the set $\{1,...,n\}$. Then, with the above notations
\begin{eqnarray}
\int \psi_{j_{1}}\wedge \cdots \wedge \psi_{j_{p}}\wedge \bar{\psi}_{k_{1}}\wedge \cdots \wedge \bar{\psi}_{k_{q}}d\mu_{C}(\bar{\Psi},\Psi) = \begin{cases}
\displaystyle 0 \quad \mbox{if $p \neq q$} \\
\displaystyle \operatorname{det}C_{J\times K} \quad \mbox{if $p=q$} 
\end{cases}
,\label{GA33}
\end{eqnarray}
with $C_{J\times K} :=(C_{ij})_{i\in J, j \in K}$.
\end{prop}

\section{Norms and Bounds for the Correlation Functions}\label{norms}

Given a skew-symmetric invertible matrix $A = (A_{ij}) \in M_{2n}(\mathbb{C})$, we define its norm by
\begin{eqnarray}
\|A\| \equiv \|A\|_{1,\infty} := \sup_{i}\sum_{j=1}^{2n}|A_{ij}|. \label{NOR1}
\end{eqnarray}
Of course, the introduction of this norm allows us to obtain bounds
for correlation functions, as follows \footnote{ The eigenvalues of $A$, or any
  of its principal minor $M$ of order $p$, are 
purely imaginary and occurs in pairs: $\pm i\eta _{j}$ with $\eta _{j}>0$, $
j=1,\ldots ,n$. So, $\text{Pf}
A=\sqrt{\det A}=\prod_{j=1}^{n}\eta _{j}\leq \left\Vert A\right\Vert ^{n}$
since the spectral radius $\mu _{A}=\max_{j}\eta _{j}$ of $A$ satisfies $\mu
_{A}\leq \left\Vert A\right\Vert $ for any matrix norm, including the
induced norm
adopted here. See e.g. Thm. 1 in Sec. 10.3 and resolved Exercise 9 in
Sec. 10.4 of \cite{LT}. } 
\begin{eqnarray}
|\operatorname{Pf}[A_{i_{k},i_{l}}]_{1\le k,l \le p}| = |\E_{\mu_{A}}[\psi_{i_{1}}\wedge \cdots \wedge \psi_{i_{p}}]|  \le \|A\|^{p/2}. \label{NOR2}
\end{eqnarray} 
Unfortunately, these estimates are not suitable for typical applications. This merely reflects that our hypotheses about the matrix $A$ are still too weak, a result of the fact that important features of a typical fermionic covariance are not being taken into account. To strengthen our hypotheses on the covariance, we revisit the discussion in the work of Brydges and Wright \cite{BryWri} and consider $A$ to be the block matrix (\ref{GA29.1}) and $C$ to be the difference $C = C^{+} - C^{-}$ between two positive-definite matrices $C^{\pm} \in M_{n}(\mathbb{C})$ such that $\det C \neq 0$. The following result is a restatement of Lemma 2.1 and Proposition 2.2 from their work.

\begin{prop}\label{estimatescorrelations} Let $J = \{i_{1},...,i_{p}\}$ be a nonempty subset of the set $\{1,...,2n\}$ and set
\begin{eqnarray}
\Lambda_{J} := \{i \in \{1,...,n\}: \hspace{0.1cm} \mbox{$i \in J$ or $i + n \in J$}\}. \label{NOR3}
\end{eqnarray}
If we set
\begin{eqnarray}
\Psi_{J} := \psi_{i_{1}}\wedge \cdots \wedge \psi_{i_{p}}, \label{NOR4}
\end{eqnarray}
then
\begin{eqnarray}
|\E_{\mu_{A}}[\Psi_{J}]| \le (4\max_{\pm, i \in \Lambda_{J}}C^{\pm}_{ii})^{\frac{|J|}{2}} \label{NOR5}
\end{eqnarray}
where $|J|$ denotes the cardinality of $J$. 
\end{prop}

Proposition \ref{estimatescorrelations} provides suitable bounds for correlation functions which are going to play a key role in the remaining of our work. Please note, however, that the assumption that $C$ can be decomposed as a difference $C = C^{+}-C^{-}$ of positive covariances $C^{\pm}$ is not essential to develop the method presented in this paper. This assumption was introduced in the original paper \cite{BryWri} and it simplifies the analysis because it immediately leads to Gram type bounds for $\operatorname{Pf}A_{J \times J} = \det C_{J \times J}$, as stated in Proposition \ref{estimatescorrelations}. The applications we are interested here do satisfy this hypothesis, so we conveniently adopted it as well. The essential point in all applications is that the expectation of a monomial of degree $j=|J|$ is bounded by $K^j$, where the constant $K$ is uniform in the number of Grassmann generators $2n$. This is the case when the covariance has a finite Gram constant or, more generally, a finite determinant constant \cite{PeSa}. See also Remark \ref{boundremark}.

Once we have set a way of estimating the contribution of a covariance matrix, the next step is to introduce norms on the underlying Grassmann algebra. For this matter, suppose we are given an element $f$ of $\bigwedge V$ whose representation in terms of its generators is 
\begin{eqnarray}
f(\Psi) = \sum_{p=0}^{\infty}\sum_{i_{1}< \cdots < i_{p}}\zeta_{i_{1},...,i_{p}}(\psi_{i_{1}}\wedge \cdots \wedge \psi_{i_{p}}). \nonumber
\end{eqnarray}
Motivated by \cite{BryWri}, we write this expression more compactly as
\begin{eqnarray}
f(\Psi) = \sum_{J\subset \{1,...,2n\}}\zeta_{J}\Psi_{J}, \label{NOR6}
\end{eqnarray}
where the sum ranges over every ordered subset $J= \{i_{1}<\cdots < i_{p}\}$ of $\{1,...,2n\}$ and, to each such subset, $\zeta_{J} := \zeta_{i_{1},...,i_{p}}$ and $\Psi_{J}$ is given as in (\ref{NOR4}). Under these notations, we also set
\begin{eqnarray}
\frac{\partial}{\partial \Psi_{J}} := \frac{\partial^{|J|}}{\partial \psi_{i_{p}}\cdots \partial \psi_{i_{1}}}, \label{NOR7}
\end{eqnarray}
so that
\begin{eqnarray}
\frac{\partial f(\Psi)}{\partial \Psi_{J}}\bigg{|}_{\Psi = 0} \equiv \zeta_{J}. \label{NOR8}
\end{eqnarray}

\begin{defi}\label{normsgrassmann} Let $z \in \mathbb{R}\setminus\{0\}$ be fixed. The norm of $f \in \bigwedge V$ is defined to be
\begin{eqnarray}
\|f\|_{z} := \sum_{m=1}^{\infty}f_{m}z^{2m}, \label{NOR9}
\end{eqnarray}
where, for each $m$, $f_{m}$ is given by
\begin{eqnarray} \label{NOR10}
f_{m} := \sup_{i\in \{1,...,2n\}}\frac{1}{2m}\sum_{J \ni i, |J| = 2m}\bigg{|}\frac{\partial f(\Psi)}{\partial \Psi_{J}}\bigg{|}_{\Psi=0}\bigg{|} \equiv \sup_{i\in \{1,...,2n\}}\frac{1}{2m}\sum_{J \ni i, |J| = 2m}|\zeta_{J}|.
\end{eqnarray}
The real number $z$ is called a norm parameter.
\end{defi}

It is important to note that, although we are strictly following the developments of the theory presented in \cite{BryWri,SW}, Definition \ref{normsgrassmann} is the first point at which our work diverges from theirs, as the parameter $z$ is introduced in the norm only with even powers. Of course, this can be done because all physical objects addressed in this work belong to the subspace of even elements with zero constant field. In truth, this is the reason why we even call $\|\cdot\|_{z}$ a norm, given that it is actually a seminorm in the full Grassmann algebra. 

\section{The Renormalization Group Transformation and Polchinski's Equation}\label{rgtransformation}

Suppose $V$ and $U$ are complex vector spaces with basis given, respectively, by $\mathcal{V} = \{\psi_{1},...,\psi_{2n}\}$ and $\mathcal{U} = \{\theta_{1},...,\theta_{2n}\}$ and let $A = (A_{ij}) \in M_{2n}(\mathbb{C})$ be skew-symmetric and invertible. In addition, let $f$ be an even $\psi$-homogeneous element of $\bigwedge(V\oplus U)$, meaning that it is even when seen as an element of $\bigwedge V$. 
The \textit{effective action} is defined as 
\beq\label{RGT1}
\tilde{F}(\Psi) := -\log [(\mu_{A}*e^{-f})(\Psi)] \; .
\eeq
Because our Grassmann algebra is finitely generated and $A$ is invertible, $F$ is well-defined for any such $f$: the expansion of $e^{-f}$ reduces to a finite sum due to the nilpotency of the Grassmann variables, and $\mu_A * f^p$ is well-defined for any $p\in \bN_0$. Moreover, a possible constant term $f_\emptyset$ in the Grassmann polynomial $f$ can be taken out of the integral, and $\mu_A * 1 = 1$ by normalization of the Grassmann Gaussian measure. Thus the logarithm is again given by an expansion that truncates to a polynomial, hence well-defined. As a consequence, $F$ is an analytic function of $A$ and $f$. 

The basic strategy of the renormalization group (RG) is to study the \textit{effective action}, not by performing this integral at once but, instead, by exploring the covariance splitting formula (\ref{GA28.2}) to perform a step-by-step integration. At each step, fluctuations are integrated out as a certain range of energy is considered. As discussed in \cite{BryKen,BryWri}, the problem is then reduced to the study of the RG transformation
\begin{eqnarray}
T_{A}: f \mapsto (T_{A}f)(\Psi) :=  -\log [(\mu_{A}*e^{-f})(\Psi)] \label{RGT2}
\end{eqnarray}
which, as in the bosonic case, satisfies a semi-group property
\begin{eqnarray}
(T_{A_{1}+A_{2}}f)(\Psi) = [T_{A_{1}}(T_{A_{2}}f)](\Psi). \label{RGT3}
\end{eqnarray}
Hence, the RG transformation induces a dynamics for the evolution of $f$, called the \textit{bare action} from now on, as the system changes with different energy scales. Following the notations and ideas in \cite{BryWri}, we would like to implement a continuous scale decomposition of the covariance matrix $A$, in which case the evolution law of the RG transformation is governed by a \textit{flow equation}, that is, a partial differential equation (PDE) referred to as the \textit{Polchinski equation}. Thus, the basic hypothesis of our analysis is that $A$ is a $C^1$ function of a parameter $t$, hence can be decomposed as an integral
\begin{eqnarray}
A = \int_{t_{0}}^{T}\dot{A}(\tau)d\tau \label{RGT4}
\end{eqnarray}
between a lower scale $t_{0}$ and an upper scale $T > t_{0}$. Here, $\dot{A}(\tau)$ denotes the derivative of $A$ with respect to $\tau$. Because the effective action is analytic in $A$, it is also $C^1$ in $t$. Thus existence is proven, and we shall in the following be concerned with bounds that are stable under taking $n \to \infty$, and useful for field theoretic applications. 

Without loss of generality, fix $t_{0} = 0$, $t \in [0,T]$ and let us write (\ref{RGT4}) as a sum
\begin{eqnarray}
A = A(t) + A_{[t,T]}, \label{RGT5}
\end{eqnarray}
with
\begin{eqnarray}
A(t) := \int_{0}^{t}\dot{A}(\tau)d\tau \quad \mbox{and} \quad A_{[t,T]} := \int_{t}^{T}\dot{A}(\tau)d\tau. \label{RGT6}
\end{eqnarray}
If $A(t)$ is assumed to be a block matrix of the form (\ref{GA29.1}) satisfying the properties in section \ref{norms}, we may define
\begin{eqnarray}
\sigma_{(s,t)}^{2} := \int_{s}^{t}(4\max_{\pm, i \in \Lambda_{J}}\dot{C}_{ii}^{\pm}(\tau))d\tau, \label{RGT7}
\end{eqnarray}
for $0 \le s \le t$. Note that our definition of $\sigma_{(s,t)}^{2}$ differs from the corresponding object in \cite{BryWri}, as our $\sigma_{(0,t)}^{2}$ is homogeneous of degree one, in the sense that a change $t \mapsto t C(\tau)$ leads to a change $\sigma_{(s,t)}^{2} \mapsto t \sigma_{(s,t)}^{2}$. This is in agreement with the discussion in \cite{SW} why the square of the Gram constant must satisfy the aforementioned homogeneity condition. Additionally, from Proposition \ref{estimatescorrelations} we get
\begin{eqnarray}
|\E_{\mu_{A_{[s,t]}}}[\Psi_{J}]|  \le \bigg{|}\bigg{(}4\max_{\pm, i\in \Lambda_{J}}\int_{s}^{t}\dot{C}^{\pm}_{ii}(\tau)d\tau \bigg{)}\bigg{|}^{\frac{|J|}{2}} \le \bigg{|}\int_{s}^{t}(4\max_{\pm,i\in \Lambda_{J}}\dot{C}^{\pm}_{ii}(\tau)d\tau)\bigg{|}^{\frac{|J|}{2}} = (\sigma_{(s,t)}^{2})^{\frac{|J|}{2}}. \label{RGT8}
\end{eqnarray}

\begin{rem}\label{boundremark} As already discussed above, our method works whenever a uniform Gram bound $\gamma_{\dot{C}(\tau)}$ for $\dot{C}(\tau)$ implies a bound $\delta_{C_{[s,t]}}^{2}$ for $C_{[s,t]}$. This condition is enough to guarantee estimate (\ref{RGT8}), even when the covariance $C_{[s,t]}$ cannot be represented as a difference $C^{+} - C^{-}$ of positive matrices,
and the right hand side of  (\ref{RGT8}) is then simply replaced by $(\delta_{C_{[s,t]}}^{2})^{\frac{|J|}{2}}$. 
\end{rem}

By utilizing the semi-group property of the RG transformation on the decomposition (\ref{RGT5}), we can redirect our attention to the scale-dependent functions
\begin{eqnarray}
(T_{A(t)}f)(\Psi) := \tilde{F}(t,\Psi) \quad \mbox{and} \quad \tilde{\varphi}(t,\Psi) := e^{-\tilde{F}(t,\Psi)} =  (\mu_{A(t)}*e^{-f})(\Psi). \label{RGT9}
\end{eqnarray}

We wish to highlight an important aspect of the RG transformation: it preserves parity. This property holds great significance in the analysis that follows, and thus, we present it as a formal proposition for clarity. Its proof can be found in \cite{FKTbook}, Lemma I.23 and, for this reason, it is omitted here. 

\begin{prop}\label{parity} If $f \in \bigwedge (V\oplus U)$ is even and $\psi$-homogeneous, in the sense of Definition \ref{homogeneouselement}, then $F(t,\Psi)$ is an even element of $\bigwedge V$ for every $t \in [0,T]$.
\end{prop}

The fact that $\tilde{\varphi}$ is given by a convolution allows us to use the Heat Kernel Formula
\begin{eqnarray}
\tilde{\varphi}(t,\Psi) = e^{\frac{1}{2}\Delta_{A(t)}}f(\Psi) \label{RGT10}
\end{eqnarray}
to prove that it satisfies the following \textit{heat equation}
\begin{align}
\label{RG11}
\begin{split}
\begin{cases}
\displaystyle\frac{\partial \tilde{\varphi}}{\partial t} = \frac{1}{2}\Delta_{\dot{A}(t)}\tilde{\varphi} \\
\displaystyle \tilde{\varphi}(0,\Psi) = e^{-f(\Psi)}
\end{cases}
\end{split}
\end{align}
whilst $\tilde{F} = -\log \tilde{\varphi}$ satisfies a nonlinear flow equation
\begin{align}
\label{RGT12}
\begin{split}
\begin{cases}
\displaystyle\frac{\partial \tilde{F}}{\partial t} = \frac{1}{2}\Delta_{\dot{A}(t)}\tilde{F}-\frac{1}{2}\langle \nabla \tilde{F}, \dot{A}(t)\nabla \tilde{F}\rangle  \\
\displaystyle \tilde{F}^{(0)}(\Psi) := \tilde{F}(0,\Psi) = f(\Psi)
\end{cases}
\end{split}
.
\end{align}
The gradient vector $\nabla \tilde{F}$ in (\ref{RGT12}) is defined in the standard way
\begin{eqnarray}
\nabla \tilde{F} := \begin{pmatrix}
\frac{\partial \tilde{F}}{\partial \psi_{1}} \\
\vdots \\
\frac{\partial \tilde{F}}{\partial \psi_{2n}} 
\end{pmatrix}
.\label{RGT13}
\end{eqnarray}

It is more natural to work with normalized functions $F$ and $\varphi$ instead of the unnormalized ones. This is done by setting
\begin{eqnarray}
\varphi(t,\Psi) := \frac{(\mu_{A(t)}*e^{-f})(\Psi)}{\E_{\mu_{A(t)}}[e^{-f}]} =: e^{-f(t,\Psi)}, \label{RGT14}
\end{eqnarray}
in which case we have $\varphi(t, 0) = 1$ and $f(t,0) = 0$ for every $t$. The following result then is an immediate consequence of our previous analysis.

\begin{theo}\label{polchinskiequation} The normalized effective action $f(t,\Psi)$ satisfies the flow equation
\begin{eqnarray}
\frac{\partial f}{\partial t} = \frac{1}{2}\Delta_{\dot{A}(t)}f - \frac{1}{2}\langle \nabla f, \dot{A}(t)\nabla f\rangle -\frac{1}{2}\Delta_{\dot{A}(t)}f\bigg{|}_{\Psi = 0} \label{RGT15}
\end{eqnarray}
with initial condition $f(0,\Psi) = f(\Psi)$.
\end{theo}

Equation (\ref{RGT15}) is called \textit{Polchinski's equation} after J. Polchinski \cite{Polch}, who first used it to prove renormalizability\footnote{To see how this is done for $\phi^{4}$ theories, see e.g. \cite{MSbook,KKS} and references therein.}.  In principle, if one solves this equation for $f(t,\Psi)$, one gets all the information about the evolution of the bare action. This is why this approach, and equations derived from it, are sometimes called the \textit{exact renormalization group}. But in most cases, finding a global solution of this equation is too difficult a task and all one can do is to prove local solvability. The idea behind the majorant method, which will be discussed in the next section, is to use a Cauchy-Kowalewski type of argument to control the solution of Polchinski's equation by the solution of an associated Hamilton-Jacobi equation, in close analogy to what is done for bosons in \cite{BryKen,BryWri}.
  
 The effective action is an even element of $\bigwedge V$, so its norm is  a polynomial in $z^{2}$, i.e.
\begin{eqnarray}\label{RGT16}
\|f(t)\|_{z} = \sum_{m=1}^{n}F_{m}(t)z^{2m} 
\end{eqnarray}
for each fixed $t$, where $F_m (t)$ is given by (\ref{NOR10}), with $f_m$ replaced by $f_m(t)$. 
In particular, for $t=0$
\begin{eqnarray}
\|f\|_{z} \equiv \|F^{(0)}\|_{z} = \sum_{m=1}^{n}F_{m}^{(0)}z^{2m} \; .\label{RGT17}
\end{eqnarray}

 For all $k \le n$, the norm coefficients $F_k (t)$ of the effective action are continuous, nonnegative functions of $t$, and $F_0 =0$. 
We set $F_k (t)  = 0$ for $k > n$,  denote the sequence of coefficients by $F(t) = (F_k (t))_{k\in\bN_0}$, and will refer to $F$ as a sequence of length $n$. Note that for $t> 0$, the $F_k$ depend on $n$. We do not write this explicitly to avoid overloading notation, but it is important to keep it in mind, because a main goal in the following is to give $n$-independent bounds for the $F_k$.

\begin{prop}\label{Fseqlem}
For all $t \ge 0$ and $k \in \mathbb{N}$, the norm coefficients of the effective action satisfy
\beq\label{Fseqineq} 
F_{k}(t) 
\le
\Fphz_k (t) 
+
 \sfrac12 \BQ_k (F,F) (t)
\eeq
with
\beq\label{F0kdef}
\Fphz_k (t) 
=
\sum_{m=k}^{n}F_{m}^{(0)}\binom{2m}{2k}\sigma_{(0,t)}^{2(m-k)}
\eeq
and (for $G$ another sequence of finite length)
\beq\label{BQdef}
\BQ_k (F,G) (t)
=
\int_{0}^{t}ds \;
\Vert \dot{A}(s)\Vert \;
\sum_{\substack{l,m\ge 1 \\ l+m \ge k+1}} 
F_{l}(s)\; G_{m}(s)\; \Gamma_{l,m}(\sigma_{(s,t)}) \; ,
\eeq
where
\begin{eqnarray}\label{RGT19}
\Gamma_{l,m}(\xi) := 4lm \sum_{\substack{\textup{$k'$, $k''$ odd} \\k'+k'' = 2k}}\binom{2l-1}{k'}\xi^{2l-1-k'}\binom{2m-1}{k''}\xi^{2m-1-k''} 
=
\xi^{2(l+m-1-k)} \; \Gamma_{l,m} (1) \; .
\end{eqnarray}
The summation range for $l$ and $m$ in (\ref{BQdef}) is infinite, but the sum for $\BQ_k (F,F)$ remains a finite sum because $F$ is of length $n$. 
\end{prop}

The proof of Proposition \ref{Fseqlem} closely follows the proof in \cite{BryWri}. We include it in Appendix \ref{proofofproposition} because this proposition plays a central role in the construction of a majorant for the effective action, and because it is crucial for our purposes to follow evenness properties of the coefficients in detail in the proof. 

The finite sequence $F(t) = (F_k (t))_{k\in \{0, \ldots, n\}}$ of norm coefficients has generating function
\beq\label{Ftzdef}
F(t,z)
=
\sum_{k=0}^n F_k (t) z^{2k}
\eeq
which is a polynomial in $z$. For $z > 0$, it coincides with the norm (\ref{RGT16}).
The finite length $n$ of the sequence $F$ corresponds to the degree $2n$ of its generating polynomial. 

\begin{lemma}\label{Fgenlem}
For all $t \ge 0$ and all $k \in \{0, \ldots, n\}$, $F$ is an even polynomial  of degree $2n$ in $z$, and for $z \ge 0$
\beq\label{Fineq}
F (t,z) 
\le
F^{(0)} (t,z) 
+
\frac12 \int_0^t \dd s\; 
\left(\frac{\del}{\del z} \frac12 \left[ F(s,z+ \sigma_{(s,t)})+ F(s,z- \sigma_{(s,t)})\right]\right)^2
\eeq
with
\beq\label{Fzerdef}
F^{(0)} (t,z)  
=
\frac12 \left( F^{(0)} (z+ \sigma_{(0,t)})  + F^{(0)} (z - \sigma_{(0,t)}) \right) \; .
\eeq
\end{lemma}

\begin{proof}
$F(t,z)$ is an even polynomial in $z$ by its definition (\ref{Ftzdef}). Multiplication with $z^{2k}$ and summation over $k \in \{ 1, \ldots, n\}$ gives
\beq
\begin{split}
\sum_{k=0}^{n}\sum_{m=k}^{n}F_{m}^{(0)}\binom{2m}{2k}\sigma_{(0,t)}^{2(m-k)}z^{2k} 
&= 
\sum_{m=1}^{k}F_{m}^{(0)}\sum_{k=0}^{m}\binom{2m}{2k}\sigma_{(0,t)}^{2(m-k)}z^{2k} 
\\
&= 
\frac{1}{2}\sum_{m=1}^{n}F_{m}^{(0)}[(\sigma_{(0,t)}+z)^{2m}+(\sigma_{(0,t)}-z)^{2m}], 
\end{split}\label{MM19}
\eeq
and similarly
\beq
\begin{split}
&\sum_{k=1}^{n}\sum_{\substack{l,m \\ l+m \ge k+1}} F_{l}(s)F_{m}(s)\Gamma_{l,m}(\sigma_{(s,t)})z^{2k} \\
&=  \sum_{1 \le l,m \le n}4lm F_{l}(s)F_{m}(s)\sum_{\substack{1 \le k' \le 2l-1 \\ 1 \le k'' \le 2m-1 \\ \textup{$k'$, $k''$ odd}}}\binom{2l-1}{k'}\binom{2m-1}{k''}\sigma_{(s,t)}^{2l-1-k'}z^{k'}\sigma_{(s,t)}^{2m-1-k''}z^{k''}  \\
&= \bigg{(}\sum_{m=1}^{n}2mF_{m}(s)\sum_{\textup{$k$ odd}}\binom{2m-1}{k}\sigma_{(s,t)}^{2m-1-k}z^{k}\bigg{)}^{2} \\
&=\bigg{(}\sum_{m=1}^{n}F_{m}(s) \frac{1}{2}\frac{\partial}{\partial z}\bigg{[}(\sigma_{(s,t)}+z)^{2m}+(\sigma_{(s,t)}-z)^{2m}\bigg{]}\bigg{)}^{2} \; .
\end{split}
\label{MM20}
\eeq
\end{proof}

Note that because $\sigma_{(t,t)} = 0$, the left hand side of the inequality can be written as the same linear combination that also appears on the right hand side: $F(t,z) = \frac12 (F(t, z+ \sigma_{(t,t)}) + F(t,z-\sigma_{(t,t)}))$. Moreover, the structure of (\ref{Fineq}) implies that for polynomials $F$ of degree $g>2$ in $z$ that have only positive coefficients, the inequality must by strict for $z > 0$ because the degree of the polynomial on the right hand side is $2g-2$.

\section{The Majorant Method}\label{majorantmethod}

The sequence $(F_k (t))_{k \in \bN_0}$ introduced above provides norms for the coefficients of the effective action for all $k \le n$. As defined the $F_k$ depend on the number of generators, $2n$, of the Grassmann algebra. Our aim is to show $n$-independent bounds for the $F_k$. The strategy for this is to replace the inequalities by equalities in Lemmas \ref{Fseqlem} and \ref{Fgenlem}, and to prove existence, uniqueness and good properties of the solution to the resulting equations, for sufficiently small $t>0$. This yields the following theorem. 

\begin{theo}\label{newmajorant}
Let the polynomial $\Fphz (z)$ be independent of $n$, i.e.\ its degree is independent of $n$ and the coefficients do not depend on $n$.  Then there is $t_0 > 0$ such that for all $t \in [0,t_0]$, a solution of 
\beq\label{phi-eq}
\phi_{k}(t) 
=
\Fphz_k (t) 
+
\frac{1}{2} \int_{0}^{t}ds\;
\Vert\dot{A}(s)\Vert
\sum_{\substack{l,m=1 \\ l+m \ge k+1}}^\infty
\phi_{l}(s)\; \phi_{m}(s)\; \Gamma_{l,m}(\sigma_{(s,t)}), 
\eeq  
exists. For all $k$, $t \mapsto \phi_k(t)$ is continuous and increasing in $t$, and 
\beq\label{F-le-phi}
0 \le F_k (t) \le \phi_k (t)
\eeq
holds for all $k \in \bN_0$ and all $t \in [0,t_0]$. Because $n$ does not appear in (\ref{phi-eq}), (\ref{F-le-phi}) provides $n$-independent bounds for the $F_k$ on the interval $[0,t_0]$. 
\end{theo}

The hypotheses on $\Fphz$ are fulfilled in all fermionic models that have a short-range polynomial interaction of fixed degree, in particular for the quartic interactions that we discuss in Section \ref{quarticapplications} below. Our analysis extends to analytic initial conditions, see Remark \ref{phiblah} below, but some proofs are simpler under the hypotheses we make in Theorem \ref{newmajorant}. 

The existence of a solution of (\ref{phi-eq}) is not obvious because (\ref{phi-eq}) is a nonlinear equation involving infinite series, without an apparent recursive structure in $k$. We prove Theorem \ref{newmajorant} in two steps:
we first prove by iteration that a solution to (\ref{phi-eq}) with the properties stated in the theorem provides the bound (\ref{F-le-phi}), then show existence and uniqueness of the solution of (\ref{phi-eq}).

We construct a sequence of functions $(\Ftn{\ell} (s))_{\ell \in \bN_0}$ as follows. Using the notation (\ref{BQdef}), define $\Ftn{0}_k (s) = F_k (s)$ and 
\beq
\Ftn{1}_k (t) 
=
F^{(0)}_k (t) 
+
\frac12 \BQ_k (\Ftn{0},\Ftn{0})(t)
\eeq
That is, $\Ftn{1}$ is the right hand side of the inequality (\ref{Fseqineq}), so , clearly, $\Ftn{1}_k (s) \ge F_k (s)$ for all $s\ge 0$ and all $k \in \bN_0$. Moreover, $t \mapsto \Ftn{1}_k (t)$ is $C^1$ by our hypotheses on the covariance, and increasing in $t$ because the same holds for $\sigma_{(0,t)}$ and because the integrand is a nonnegative function. Also, $\Ftn{1}$ is still a sequence of finite length, but its length is now $2n-1 \le 2n $. Set 
\beq
\Ftn{2}_k (t) 
=
F^{(0)}_k (t) 
+
\frac12 \BQ_k (\Ftn{1},\Ftn{1})(t)
\eeq
Because $\BQ_k$ is increasing in $F$, $\Ftn{2}_k (s) \ge \Ftn{1}_k (s)$ for all $s\ge 0$ and all $k \in \bN_0$. Moreover, $t \mapsto \Ftn{2}_k (t)$ is increasing in $t$,and $\Ftn{2}_k- F^{(0)}_k (t)$ is $C^2$ and convex because it is the integral of an increasing $C^1$ function. $\Ftn{2}$ is a sequence of length $4n-3 \le 4n$. Continuing the recursion by 
\beq
\Ftn{\ell+1}_k (t) 
=
F^{(0)}_k (t) 
+
\frac12 \BQ_k (\Ftn{\ell},\Ftn{\ell})(t)
\eeq
then gives an increasing sequence of sequences of increasing length and increasing differentiability: 
$\Ftn{\ell}_k- F^{(0)}_k (t)$ is $C^\ell$, $\Ftn{\ell}_k (s) \le \Ftn{\ell+1}_k (s) $.

For every $t \ge 0$ and $k\in \bN_0$, $(\Ftn{\ell}_k (t))_{\ell \in \bN_0}$ is an increasing sequence of positive numbers. Thus it converges if it is bounded above by an $\ell$-independent constant. 

\begin{lemma}
If there is $t_0 > 0$ so that $(\Ftn{\ell}_k (t_0))_{\ell \in \bN_0}$ converges for all $k$, 
then the sequence $(\Ftn{\ell}_k (t))_{\ell \in \bN_0}$ converges for all $t \in [0,t_0]$ and all $k \in \bN_0$. Indeed, convergence is uniform on $[0,t_0]$, and the limiting functions $\phi_k (t) = \lim_{\ell \to \infty} \Ftn{\ell}_k (t)$ are $C^\infty$ in $t$ and satisfy $(\ref{phi-eq})$. 
\end{lemma}

\begin{proof}
Let $k \in \bN_0$. Set $\phi_k (t_0) = \lim_{\ell \to \infty} \Ftn{\ell}_k (t_0)$. Because $(\Ftn{\ell}_k (t_0))_{\ell \in \bN_0}$ is increasing, $\Ftn{\ell}_k (t_0) \le \phi_k (t_0)$ for all $k$. Because $t \mapsto \Ftn{\ell} (t) $ is increasing in $t$, $\Ftn{\ell}_k (t) \le \phi_k (t_0)$ for all $t \le t_0$, too. Thus the increasing sequence $(\Ftn{\ell}_k (t))_{\ell \in \bN_0}$ converges for all $k$ and all $t \in [0,t_0]$. Moreover
\beq\label{unifrate}
0 \le \Ftn{\ell+1 +m}_k (t) - \Ftn{\ell+1}_k (t) 
= 
\frac12 \BQ_k (\Ftn{\ell +m} - \Ftn{\ell}, \Ftn{\ell+m}) (t)
+
 \frac12 \BQ_k ( \Ftn{\ell}, \Ftn{\ell +m} - \Ftn{\ell}) (t) \; .
\eeq
Because the only $t$-dependence of $\BQ_k$ comes from the upper boundary of the integral, and because $\Ftn{\ell +m}_k (s)  - \Ftn{\ell}_k (s) \ge 0$ everywhere, (\ref{unifrate}) implies that $t \mapsto \Ftn{\ell +m}_k (t) - \Ftn{\ell}_k (t)$ is increasing in $t$ for all $\ell$. Thus 
\beq
\sup_{t \le t_0} \left[ \Ftn{\ell +m}_k (t) - \Ftn{\ell}_k (t) \right] \le \Ftn{\ell +m}_k (t_0) - \Ftn{\ell}_k (t_0)
\eeq
and convergence is uniform on $[0,t_0]$. A slight extension of this argument implies convergence of derivatives.

The infinite series in (\ref{BQdef}) have all nonnegative terms, the convergence of  $\Ftn{\ell}_k (t)$ is monotone, and the result is bounded by $\phi_k (t_0)$. Thus the monotone convergence theorem applies to the sums and integrals in $\BQ_k$, and it implies that the limit of the sequence can be taken inside, which implies that the limit fulfils (\ref{phi-eq}).
\end{proof}

\begin{rem}\label{phiblah}
Let $\phz = (\phz_k)_{k \in \bN_0}$ be a sequence such that $\Fphz_k \le \phz_k$ for all $k$, and that there is $R>0$ such that 
\beq
z \mapsto \phz (z) = \sum_{k=0}^\infty \phz_k \; z^{2k}
\eeq
defines an analytic function on $\{ z \in \bC : |z| < 2R\}$. Let $\delta > 0$ be so small that for all $t \le \delta $: $ \sigma_{(0,t)} < R$. For these $t$ set 
\beq
\phz (t,z)  
=
\frac12 \left( \phz (z+ \sigma_{(0,t)})  + \phz (z - \sigma_{(0,t)}) \right) \; .
\eeq
Then the analogue of (\ref{phi-eq}) in which $\Fphz_k (t) $ is replaced by $\phz_k (t)$ has a solution on a sufficiently short interval $[0,t_0']$. 
\end{rem}

The remaining task is to prove convergence at some $t_0 > 0$. We first give a brief motivation for how we do the proof. The above iteration has shown that a sequence of functions $\phi_k (t)$ that solves (\ref{phi-eq}) cannot be of finite length.\footnote{unless it is of length 1, which is uninteresting for our purposes}
Thus the combinatorial generating function 
\beq\label{phigen1}
\phi (t,z) 
=
\sum_{k=0}^\infty \phi_k (t) \; z^{2k}
\eeq
can be defined in analogy to the norm $\Vert F(t)\Vert_z$ introduced in (\ref{RGT16}), but the parameter $z$ is now a formal parameter, i.e.\ $\phi (t,z)$ is a formal power series. Then the same summations as in the proof of Lemma \ref{Fgenlem}, imply 
\beq\label{majeq}
\phi(t,z) = 
F^{(0)} (t,z)
+ 
\frac12 
\int_0^t \dd u\; 
\|\dot{A}(u)\|\bigg{(}\frac{1}{2}\frac{\partial}{\partial z}\bigg{[}\phi(u,z+\sigma_{(u,t)})+ \phi(u,z-\sigma_{(u,t)})\bigg{]}\bigg{)}^{2} \; .
\eeq
This step is not completely rigorous because formal parameters cannot be shifted by real variables in the framework of formal power series. One can avoid this problem by multiplying $\sigma_{(s,t)}^2$ by a formal expansion parameter $h$ and consider a formal expansion in both $z$ and $h$. (The parameter $h$ counts loops and it is the formal analogue of $\hbar$ in quantum field theoretical loop expansions.) 

However, we will not need this formal expansion. Instead, we will construct an analytic function $\phi(t,z)$ of $z$ that satisfies (\ref{majeq}) on an interval $[0,t_0]$ with $t_0 > 0$, as well as $\phi(t,-z) = \phi(t,z)$. It then follows by comparing coefficients of $z$ in the power series that its expansion coefficients $\phi_k (t)$ satisfy (\ref{phi-eq}). Our construction also implies that the $\phi_k (t)$ are all nonnegative, that they are unique, and that they are $C^1$ functions of $t$. From this, Theorem \ref{newmajorant} follows. 

To construct $\phi$, we first note again that because $\sigma_{(t,t)} = 0$, $\phi (t,z) = \frac12 ( \phi (t,z+ \sigma_{(t,t)}) + \phi (t,z- \sigma_{(t,t)}) )$. Thus, introducing 
\beq
\psi (u,t,z) 
=
\frac12 \left[\phi(t,z+\sigma_{(u,t)})+ \phi(t,z-\sigma_{(u,t)})\right] \; ,
\eeq
(\ref{majeq}) now reads
\beq
\psi(t,t,z)
=
F^{(0)} (t,z)
+ 
\frac12 
\int_0^t \dd u\; 
\Vert\dot{A}(u)\Vert
\left(
\frac{\partial \psi(u,t,z)}{\partial z}
\right)^2 \; .
\eeq
Therefore we will get a solution of (\ref{majeq}) if, given $t>0$, we can solve the auxiliary equation
\beq\label{majaux}
\psi(s,t,z) = 
F^{(0)} (t,z)
+ 
\frac12 
\int_0^s \dd u\; 
\Vert\dot{A}(u)\Vert
\left(
\frac{\partial \psi(u,t,z)}{\partial z}
\right)^2  
\eeq
for all $s \le t$. In the auxiliary equation, the dependence on $t$ only enters via the first term, and the $s$-dependence is only in the upper integration boundary. Taking an $s$-derivative gives the initial value problem
 \beq
 \begin{split}\label{HamJacaux}
 \frac{\del \psi}{\del s} (s,t,z)
& =
 \frac12 
 \Vert\dot{A}(s)\Vert
\left(
\frac{\partial \psi(s,t,z)}{\partial z}
\right)^2   
\\
 \psi(0,t,z) 
 &= 
F^{(0)} (t,z)
 \end{split}
 \eeq
in which the $t$-dependence enters only via the initial condition. 
In Appendix \ref{new-AppendixB}, we show that there is $t_0> 0$ such that for all $t \le t_0$, a solution of (\ref{HamJacaux}) exists which is defined for all $s \in [0,t]$ and is analytic and even in $z$. We also give a lower bound on $t_0$ in terms of the initial condition $\Fphz$ and complete the proof of Theorem \ref{newmajorant}.

\bigskip
We now discuss the Hamilton-Jacobi estimate when the initial condition is given by a polynomial of degree four, which is relevant for models with four-fermion-interactions, that is
\begin{eqnarray}
\Fphz(z) = \alpha z^{4} \label{QP1}
\end{eqnarray}
for some $\alpha > 0$. The $t$-dependent initial condition for $u(z) = \frac{\del \psi}{\del z}$  reads
\begin{eqnarray}
u_{0}(z) = 12\alpha \sigma_{(0,t)}^{2}z + 4\alpha z^{3}. 
\label{QP2} 
\end{eqnarray}
Following Appendix \ref{new-AppendixB}, a solution exists if 
$|\tau(s) u'_0(0)| < \eta < 1$, that is, 
\beq\label{taubou}
|\tau(s)| < \frac{\eta}{12 \alpha \sigma_{(0,t)}^{2} }\; .
\eeq
By choosing $\eta<1$ sufficiently small, the further condition $|\rho(\tau(s),w)| < 1$, where $\rho$ is defined in (\ref{rhodef}), and $\tau$ is the function defined in (\ref{varchange}), is ensured.

\section{Application to Fermions with Local Quartic Interactions}\label{quarticapplications}

The theorems proven above imply that, when the solution of the Hamilton-Jacobi equation for the majorant exists, the fermionic effective action is analytic in the fields and in the initial interaction, uniformly in the dimension of the Grassmann algebra (i.e.\ the number of generators). To apply this to quantum field theoretical models of fermions, we first connect it to the standard formulations of such models. We then show how standard perturbative power counting bounds for a fermionic model with a $\psi^4$ interaction can be proven using the majorant method. We do this for the RG-irrelevant terms, by implementing the scale dependence properly in the majorant method. A full renormalization analysis would be beyond the scope ot this paper, but is underway. In the following we will use some of the notation that has become standard in Polchinski theory. 

For relativistic fermions or their Euclidian counterparts, the fields are Dirac spinors and the kinetic term in the action contains the Dirac operator. In order to have as little extra complications as possible, we focus here instead on the fermionic analogue of the simplest scalar field theory in the Euclidian, where the kinetic term is given by the Laplacian, and show scale-dependent bounds that are direct analogues of those proven, e.g.\ in \cite{Polch,KKS} for scalar field theory itself. Of course, a fermionic model with such a Laplacian as kinetic term is not reflection-positive, so it will not satisfy the Osterwalder-Schrader axioms.  

\subsection{Regularized $\Psi^{4}_{d}$ Theory}\label{phi4d}
There are many ways to regularize formal functional integrals of quantum field theory to get a mathematically well-defined integral as a starting point, and to introduce a scale-dependent flow. Here, as in \cite{Sa98}, we choose a double regularization: we replace space ${\mathbb R}^d$ by a finite lattice ${X}$, and also introduce a momentum space regulator at a large scale $\Lambda_0 > 0$ that removes singularities in the covariance at coinciding points, in the usual way employed in the Polchinski setup \cite{Polch, KKS, Sa98}. The finite lattice is introduced to make the field algebra and the integrals finite-dimensional and obviously well-defined. The renormalization problem for the model is about the $\Lambda_0$-dependence. 

The procedure is then to show first that the infinite-volume limit $L \to \infty$ and the continuum limit $\veps \to 0$ of the effective action exist at fixed $\Lambda_0$, hence defining a UV-regularized, infinite-volume effective action that has full Euclidian symmetry, and then to take the limit $\Lambda_0 \to \infty$. 

At fixed $\Lambda_0 < \infty$ and $\mass > 0$, proving convergence of the effective action in the limits $L \to \infty$ and $\veps \to 0$ is not difficult, and similar to the proof given in \cite{Sa98} (see below). Proving that the limit $\Lambda_0 \to \infty$ exists requires placing counterterms, i.e.\ posing a $\Lambda_0$-dependent initial condition for the interaction and showing that the $\Lambda_0$-dependence can be chosen such that the effective action has a finite limit. Polchinski's proof \cite{Polch}, in the form of \cite{KKS}, achieves this in a formal perturbation expansion by introducing a flow parameter $\Lambda = \Lambda_s = \Lambda_0 e^{-s}$ which effectively imposes an infrared cutoff, and taking $s$ from zero to infinity, corresponding to a successive integration over fields and a corresponding RG flow of effective actions (and hence correlation functions). 

Let $\alatt > 0$ and $L >0$ such that $L/\alatt$ is integer, and let $X$ be the discrete torus $X = \alatt \bZ^d / L \bZ^d$. That is, $X$ is a finite hypercubic lattice ${X}$ of spacing $\alatt$ and sidelength $L$ with periodic boundary conditions To each lattice point $x$, we associate Grassmann generators $(\psq_\alpha (x), \ps_\alpha (x))_{\alpha \in \{1, \ldots, D\}}$. Here $D\ge 2$ is a fixed integer (the dimension of the spinor). We require $D \ge 2$ so that a nonvanishing quartic term local in $x$ can be formed by the nilpotent Grassmann variables. Denote ${\bX} = \{1, \ldots , D\}  \times X$.
The action of the Euclidian $\Psi^4$ model reads
\beq\label{psi4action}
S
=
\int_X 
\left(
\sfrac12 
{\textstyle \sum\limits_{\alpha = 1 }^D} 
\psq_\alpha (x) (-\Delta + \mass^2) \ps_\alpha (x)
+
\sfrac{\lambda}{4} \; 
(\psq\psi (x))^2
\right)
\eeq
with $\psq\psi (x) = \sum_{\alpha = 1 }^D \psq_\alpha (x) \ps_\alpha (x)$, $\Delta$ the lattice Laplacian, $\mass > 0$ the mass parameter, $\lambda > 0$ the coupling constant, and the integral denoting $\alatt^d$ times a sum over lattice points. 

The total number of Grassmann generators, $2n= 2D(L/\veps)^d$, is finite. 
Let $\nu: \bX \to \{1, \ldots, n\}$ be any enumeration of $\bX$ and set $\psi_{\nu(\alpha,x)} = \ps_\alpha (x)$ and 
$\psi_{n+\nu(\alpha,x)} = \psq_\alpha (x)$. Then we are in the framework of the previous sections, and all the results there apply.\footnote{When using this representation, we absorb the factors $\alatt^d$ in front of the lattice sums into the summands} In particular, the effective action, which generates all connected amputated Green functions, exists. Moreover, it is independent of the enumeration.

\subsection{Setup of the renormalization group flow}\label{flose}
Let $\Lambda_0 > \mass $ be a large positive number, $s \ge 0$, and $\Lambda_s = \Lambda_{0}\; e^{-s}$. 
The $s$-dependent Grassmann Gaussian covariance $A_s$ is defined as follows. 

For $p \in {\mathbb R}^d$ let 
\beq
\hat C (p) 
=
\frac{1}{\hat D_\alatt (p) +\mass^2}
\eeq
with $\hat D_\alatt (p) = \sfrac{2}{\alatt^2} \sum_{\mu =1}^d (1 - \cos (\alatt p_\mu))$. Set
\beq
\hat C_s (p)
=
\hat C (p) \; 
\left( e^{-\frac{p^{2}+\mass^2}{\Lambda_0^{2}}} - e^{-\frac{p^{2}+\mass^2}{\Lambda_s^{2}}}\right) 
\eeq
and 
\beq\label{Csdef}
C_s (x,x') 
=
L^{-d} \sum_{p} e^{\I p (x-x')}\; \hat C_s (p)
\eeq
with $p$ summed over the torus $\mathcal{B}  = \frac{2\pi}{L} {\mathbb Z}^d / \frac{2\pi}{\alatt} {\mathbb Z}^d$ 
(equivalently, over $[-\frac{\pi}{\alatt} , \frac{\pi}{\alatt})^d \cap \frac{2\pi}{L} \bZ^d$)
Evidently, $C$ is symmetric in $x$ and $x'$, and it depends only on $x-x'$. 
Let $\veps_{01} = - \veps_{10} =1$ and $\veps_{00}=\veps_{11} = 0$, and define for charge index $c \in \{0,1\}$ (where $c=0$ corresponds to $\psi$ and $c=1$ to $\psq$)
\beq\label{GA29}
(A_s)_{cn+\nu(\alpha,x), c'n+\nu(\alpha',x')}     
=
\veps_{cc'} \; \delta_{\alpha\alpha'}\; C_s (x,x') \; .
\eeq
Then, for $s\to \infty$ and $\Lambda_0 \to \infty$, $A_0$ tends to the unregularized covariance of the action (\ref{psi4action}) at $\lambda =0$. The covariance at $s=0$ vanishes, so no fields have been integrated over at $s=0$. For $s\to\infty$, the covariance tends to the $\Lambda_0$-regularized covariance without infrared cutoff, so in this limit, the fields on all scales have been integrated over.

For $t \ge s$, set 
\beq
C_{s,t} = C_s -C_t = - \int_s^t \dot C (s') \; ds' 
\eeq
where we have again denoted the $s$-derivative by a dot, and $A_{s,t} = A_s -A_t$. 

In the limit $\veps \to 0$ at fixed $\Lambda_0$ and $s$, 
\beq
\hat C_s (p)
=
\frac{1}{p^2+\mass^2}\;
\left( e^{-\frac{p^{2}+\mass^2}{\Lambda_0^{2}}} - e^{-\frac{p^{2}+\mass^2}{\Lambda_s^{2}}}\right) 
\eeq
and 
\beq
\sfrac{d}{ds}\; \hat C_s (p)
=
\sfrac{2}{\Lambda_s^2} \; 
e^{-\frac{p^{2}+\mass^2}{\Lambda_s^2}}
\eeq

\subsection{The determinant and decay bounds}
We use a Gram bound $\gamma_{s,t}$ for the covariance $\dot C_{s,t}$ to estimate $|\E_{\mu_{A_{[s,t]}}}[\Psi_{J}]| \le \gamma_{s,t}^{|J|}$, in analogy to  (\ref{RGT8}). We now briefly recall the Gram bound and calculate $\sigma_{(s,t)}^{2}$. The fermionic quadratic form corresponding to $C_{s,t}$ is $\int_{x,x'} \sum_{\alpha,\alpha'} \psq_\alpha (x) \delta_{\alpha,\alpha'} \; C_{s,t}(x,x') \; \ps_{\alpha'} (x') $. Let $e_1, \ldots e_D$ be any orthonormal basis of $\bC^D$, $\cH = \bC^D \otimes L^2 (\bR^d)$ and 
\beq
g_{x} (p) 
=
e^{-\I px}\; \sqrt{\hat C_{s,t} (p)} \; .
\eeq
Then $e_\alpha \otimes g_{x} \in \cH$ for all $\alpha$ and $x$, and 
\beq
\delta_{\alpha,\alpha'} \; C_{s,t}(x,x') = \langle e_\alpha \otimes g_{x} \mid e_{\alpha'} \otimes g_{x'} \rangle_{\cH} \; .
\eeq
Thus $C_{s,t}$ has a Gram constant $\gamma_s=\| g_{x}\|_2$ and hence
\beq\label{PHI4.8}
\gamma_{s,t}^2 = \int d^d p \; | {\hat C}_{s,t} (p) | = \|  {\hat C}_{s,t} \|_1
\eeq
Because
\begin{eqnarray}
\|\sfrac{d}{ds} \hat{C}_{s}\|_{1} = \frac{2}{\Lambda_{s}^{2}}e^{-\frac{\mass^2}{\Lambda_{s}^{2}}}\int e^{-\frac{|p|^{2}}{\Lambda_{s}^{2}}}d^dp = 2\pi^{\frac{d}{2}}e^{-\frac{\mass^2}{\Lambda_{s}^{2}}}\Lambda_{s}^{d-2}, \label{PHI4.9}
\end{eqnarray}
it follows that
\beq\label{PHI4.10}
\|  {\hat C}_{s,t} \|_1
\le 
\int_{s}^{t}ds' \|\sfrac{d}{ds'} \hat{C}_{s'}\|_{1}
=
2\pi^{\frac{d}{2}}\int_{s}^{t}ds'e^{-\frac{\mass^2}{\Lambda_{s'}^{2}}}\Lambda_{s'}^{d-2} \le 2\pi^{\frac{d}{2}}\int_{s}^{t}ds'\Lambda_{s'}^{d-2} = \rho_{d}(\Lambda_{s}^{d-2}-\Lambda_{t}^{d-2}),
\eeq
where we introduced
\begin{eqnarray}
\rho_{d} := \frac{2\pi^{\frac{d}{2}}}{d-2}. \label{PHI4.11}
\end{eqnarray}
Hence, the first parameter of our theory is found
\begin{eqnarray}
\sigma_{(s,t)}^{2} = \rho_{d}(\Lambda_{s}^{d-2}-\Lambda_{t}^{d-2}). \label{PHI4.12}
\end{eqnarray}

For the decay bound we need to obtain the norm $\|\dot{C}_{s}\|_{1,\infty}$. Because $\dot{\hat{C}} = \frac{2}{\Lambda_s^2}\; e^{- \frac{p^2 + \mass^2}{\Lambda_s^2} }$ is a Gaussian in $p$, $\dot C$ is Gaussian, too, hence nonnegative, so
\beq
\int_{\mathbb{R}^{d}}d\xx |\dot{C}_{s}(\xx)| = 
\int_{\mathbb{R}^{d}}d\xx \dot{C}_{s}(\xx) =
\dot{\hat{C}}_{s}(0) = \frac{2}{\Lambda_{s}^2}e^{-\frac{\mass^2}{\Lambda_{s}^{2}}}, \label{PHI4.13}
\eeq
we get
\begin{eqnarray}
\|\dot{C}_{s}\|_{1,\infty} \equiv \|\dot{C}_{s}\| = \frac{2}{\Lambda_{s}^2}e^{-\frac{\mass^2}{\Lambda_{s}^{2}}}. \label{PHI4.14}
\end{eqnarray}
Note that the latter vanishes exponentially fast as $\frac{\mass^2}{\Lambda_{s}^{2}} \to +\infty$. It is also uniformly bounded in $s$ because
\begin{eqnarray}
\frac{2}{\Lambda_{s}^2} \; e^{-\frac{\mass^2}{\Lambda_{s}^{2}}} \le \frac{2}{\mass^2}\sup_{\xi \ge 0}\xi e^{-\xi} \le \frac{2}{\mass^2}.\label{PHI4.15}
\end{eqnarray}

We have given these bounds for the continuum and infinite-volume limit of $C_{s,t}$ and $\dot C_s$. Similar bounds also hold for $\alatt > 0$ and $L < \infty$. The replacement in the Gram bound is that $L^2 (\bR^d)$ gets replaced with $L^2(L^{-1} \bZ^d \cap [-\frac{\pi}{\veps}, \frac{\pi}{\veps})^d)$. Because $\Lambda_s < \infty$ and $\mass > 0$, the function $p \mapsto \hat C_{s,t} (p)$ is a Schwartz function. Thus the Gram bound for $\alatt > 0$ and $L < \infty$ converges to the one calculated above for $L\to \infty$ and $\alatt \to 0$, hence is also uniformly bounded in $\alatt$ and $L$.   It is standard to prove finiteness of the decay constant using summation by parts-techniques (see, e.g.\ \cite{FKTbook} and \cite{MSbook}). We omit that proof here to keep the presentation concise. The essential point is that $\sigma_{(s,t)}^{2} $ and $\|\dot{C}_{s}\|_{1,\infty}$ are bounded uniformly in $\alatt$ and $L$, because $\Lambda_s < \infty$ and $\mass > 0$. Because the $L$ dependence is that of a Riemann sum converging to the integral of a Schwartz function, and because of the exponential decay introduced by $\Lambda_s$, the values of $\sigma_{(s,t)}^{2} $ and $\|\dot{C}_{s}\|_{1,\infty}$ are continuous in $\alatt$ and $1/L$, hence the estimates above apply with a prefactor arbitrarily close to $1$ also for sufficiently large volume and small lattice spacing. A dominated-convergence argument then shows that the limits $L\to \infty$ and $\alatt \to 0$ of the effective action 
\beq
F_0 (\phi) 
=
- \log \left(
\mu_{A_0} * e^{-I_0}
\right)
\eeq
with $I_0 (\psq,\ps) = \frac{\lambda}{4} \int_X (\psq\ps)(x)^2$
exist for small enough $|\lambda|$, and satisfy the majorant bound with the same constants as above, hence are analytic in $\lambda$. The radius of analyticity depends on $\Lambda_0$. In the next section, we give a lower bound for it for a skeleton flow, where all two-point corrections are projected out of the flow. 

\subsection{Scale-dependent bounds in $\Psi_{4}^{4}$ Theory}\label{scalingphi4}

Because the majorant for the effective action is obtained by solving a Hamilton-Jacobi equation,  one wants to have coefficient functions in this Hamilton-Jacobi equation that are of order one (\enquote{dimensionless} in the RG terminology), so they do not diverge in the limit one wants to take. We scale the $F_{m}(t)$ in (\ref{RGT16}) by a power of $\Lambda_t$, i.e.\ write 
\beq\label{tilF}
F_{m}(t)=\Lambda_{t}^{a+2bm} \tilde F_{m}(t) \; .
\eeq 
Here $a$ and $b$ are real numbers that depend only on the dimension $d$. The effect of scaling these norm coefficients is that both parameters $\sigma_{(0,t)}$ and $\|\dot{A}(t)\|$ will get rescaled, too.  The aim is to find $a$ and $b$ such that, after this rescaling, the coefficients in the Hamilton-Jacobi equation will be bounded uniformly in $t$. The values of $a$ and $b$ that we shall take are the ones suggested by the results from perturbative renormalization \cite{MSbook}. 

This scaling also means that the initial condition is taken as $\tilde F_m (0) = \tilde F_m^0$, where the $\tilde F_m^0$ are given and bounded independently of $\Lambda_0$, for $m \ge 3$. For $\tilde F_2^0$, we make the same hypothesis here, but in a complete treatment, this function may contain counterterm contributions, which depend on $\Lambda_0$. In this work, we set $\tilde F_1^0 = 0$. 

Let us show in more detail how we rescale $\sigma_{(0,t)}$ and $\|\dot{A}(t)\|$ by scaling the norm coefficients of $F(t,\Psi)$. As we know, the latter satisfies Polchinski's equation, so its norm coefficients $F_{m}(t)$ satisfy (\ref{Fseqineq}) in Proposition \ref{Fseqlem}. Let $s \le t$ and 
\beq
\tilde\sigma_{(s,t)}  = \Lambda_s^b \; \sigma_{(s,t)} \; .
\eeq
Insert (\ref{tilF}) into (\ref{Fseqineq}) and divide by the factor $\Lambda_t^{a+bk}$, to get an inequality for $\tilde F_k (t)$. The first term on the right hand side of (\ref{Fseqineq}) becomes
\begin{eqnarray}\label{linterm}
\sum_{m=k}^{n}\frac{\Lambda_{0}^{a+2bm}}{\Lambda_{t}^{a+2bk}} \; \binom{2m}{2k}\; \tilde F_{m}^{(0)}\; \sigma_{(0,t)}^{2(m-k)} 
=
\left(\frac{\Lambda_t}{\Lambda_0}\right)^{-a -2bk} \;
\sum_{m=k}^n 
\binom{2m}{2k}\;
\tilde F_m^0\;
\tilde\sigma_{(0,t)} ^{2 (m-k)}
\end{eqnarray}
By similar rearrangements, and using (\ref{PHI4.14}), the second term on the right hand side of (\ref{Fseqineq}) becomes
\beq\label{quadterm}
\int_0^t ds\; 
\left(
\frac{\Lambda_t}{\Lambda_s}
\right)^{-a-2bk} \;
e^{-\frac{\mass^2}{\Lambda_s^2}} \;
\Lambda_s^{a + 2b -2} \; 
\sum_{\substack{l,m \\ l+m \ge k+1}} \tilde F_{l}(s)\; \tilde F_{m}(s)\;
\Gamma_{l,m}(\tilde\sigma_{(s,t)}) \; .
\eeq
For all $k$ for which $a+2bk \le 0$, and all $0 \le s \le t$, $(\frac{\Lambda_t}{\Lambda_s})^{-a-2bk} \le 1$ can be used both in (\ref{linterm}) and (\ref{quadterm}), so that for these $k$,
\beq\label{goodk}
\tilde F_k (t) 
\le
\sum_{m=k}^n 
\binom{2m}{2k}\;
\tilde F_m^0\;
\tilde\sigma_{(0,t)} ^{2 (m-k)}
+
\int_0^t ds\; 
e^{-\frac{\mass^2}{\Lambda_s^2}} \;
\Lambda_s^{a + 2b -2} \; 
\sum_{\substack{l,m \\ l+m \ge k+1}} \tilde F_{l}(s)\; \tilde F_{m}(s)\;
\Gamma_{l,m}(\tilde\sigma_{(s,t)}) \; .
\eeq
We now specialize to the case $d=4$. Perturbatively, the model with a local quartic interaction is then just renormalizable, i.e.\ rendering the Schwinger functions finite in the limit $\Lambda_0 \to \infty$ requires only three types of counterterms (mass, field strength, and coupling), but new counterterms of these types are required in every order in perturbation theory \cite{KKS,MSbook}. The analysis that proves this also suggests which exponents to choose: we set $a = 4$ and $b=-1$. Then $a+2b-2=0$ and $a+2bk = 4-2k \le 0$ for all $k \ge 2$, so that leaving out $k=1$ allows us to use (\ref{goodk}). By (\ref{PHI4.12}) with $d=4$, the effect of scaling is that $\sigma_{(s,t)}$ gets replaced with 
\beq
\tilde\sigma_{(s,t)} = \Lambda_s^{-1}\; \sigma_{(s,t)} \le \sqrt{\rho_4}
\eeq
and $\|\dot A (s) \|$ gets replaced with
\beq
a(s) = e^{-\frac{\mass^2}{\Lambda_s^2}} \le 1\; .
\eeq
Thus with this choice of $a$ and $b$, rescaling has made the coefficients in the integral inequality (\ref{goodk}) uniformly bounded in $s$ and $t$. 

To prove analyticity, we cannot deal with individual values of $k$ separately but use the generating function that gives the majorant, i.e.\ sum over $k$. Since the case $k=1$ is not covered by the inequality, we can only consistently sum over $k$ if we also change the right hand side of the equation by leaving out the all terms of the sum where $l=1$ or $m=1$. In a graphical expansion, this amounts to leaving out the {\em two-point insertions}. This procedure, which corresponds to a projection $\cP_{\ge 2}$ on the field algebra where quadratic terms are mapped to zero, is standard as a first step in perturbative studies of the Polchinski equation and in constructive renormalization group. Specifically, the above-mentioned truncation means that we modify Polchinski's equation (\ref{RGT15}) by applying the projection $\cP_{\ge 2}$ to its right hand side. Including the quadratic terms requires a much more detailed analysis: the flow of the terms with $m \le 2$ has to be traced in more detail, and growing terms need to be controlled by appropriately changed initial conditions (renormalization by counterterms). We are not addressing this renormalization procedure in this work; it requires, among others, estimates for more general weighted norms and Taylor expansion arguments. 

The change of variables (\ref{varchange}) is now replaced with
\beq
t \to \tilde\tau (t)  = \int_0^t a(s) \; ds \; .
\eeq 
Since $a$ is smooth and $a(s) \ge 0$ for all $s$, $\tilde \tau$ is smooth and strictly increasing in $t$. 
Because $0 \le a(s) \le 1$, $\tilde\tau (t) \le t$. The exponential decay of $a(s)$ at large $s$ implies that $\tilde\tau$ is a bounded function:
\beq\label{ttaubou}
\tilde\tau(t) \le \frac{1}{2e} + \min\{ t, \ln \frac{\Lambda_0}{m} \}
\eeq
For $t \le t_{\Lambda_0} = \ln \frac{\Lambda_0}{m}$, this bound follows from $\tilde\tau (t) \le t$.
For $t > t_{\Lambda_0}$, split the integration interval into $[0,t_{\Lambda_0})$ and $[t_{\Lambda_0},t]$, and use the same estimate on the first interval to get
\beq
\tilde\tau (t) \le \ln \frac{\Lambda_0}{m} + \int_{t_{\Lambda_0}}^t a(s) \; ds 
\eeq
Substituting $s = t_{\Lambda_0} + \sigma$ in the integrand gives
\beq
\int_{t_{\Lambda_0}}^t a(s) \; ds
=
\int_{\Lambda_0}^{t-t_{\Lambda_0}} e^{-e^{2\sigma}} \; d\sigma
\le
\int_{\Lambda_0}^{\infty} e^{-1-2\sigma} \; d\sigma
=\frac{1}{2e} \; .
\eeq
so (\ref{ttaubou}) holds. Equation (\ref{ttaubou}) makes precise the idea that the flow ``essentially stops'' at $t_{\Lambda_0}$, i.e.\ when $\Lambda_t$ goes below $m$. 

In condition (\ref{taubou}), $\tilde\tau(t)$ and $\tilde\sigma_{(0,t)}$ now appear instead of $\tau(t)$ and $\sigma_{(0,t)}$. The coefficient $\alpha$ of the quartic interaction does not receive a scaling factor. Thus
\beq\label{PHI44.18}
|\tilde\tau(t)| < \frac{\eta}{12 \alpha \tilde\sigma_{(0,t)}^{2} }\; .
\eeq
Rearranging this as an inequality for $\alpha$,  we have proven that for
\beq
\alpha < \frac{\eta}{12 \rho_4 (1+ \ln \frac{\Lambda_0}{m})}
\eeq
the majorant converges and hence the generating function obtained from the Polchinski equation (\ref{RGT15}) with $\cP_{\ge 2}$-projected right hand side is analytic in the fields. This holds for general (not only quartic) interactions $\tilde F_m^{(0)}$, provided the norm of $\tilde F^{(0)}$ is small enough. (Recall that this means that interaction terms of degree $2m > 4$ get inverse powers of $\Lambda_0$, i.e.\ $F_m^{(0)} = \Lambda_0^{4-2m} \tilde F_m^{(0)}$.)

This bound does not allow to take the limit $\Lambda_0 \to \infty$, and we now discuss the reasons for this. At face value, the main reason is that $\tilde\tau(t)$ grows linearly in $t$, i.e.\ logarithmically in $\Lambda$, before levelling off at about $\ln \frac{\Lambda_0}{m}$, and one may wonder whether one can avoid this by doing more careful bounds. Indeed, the estimate $(\frac{\Lambda_t}{\Lambda_s})^{2k-4} \le 1$ is wasteful for $2k-4 > 0$, and a more careful bound (first integrating over $s$ and then cancelling the power of $\Lambda_t$) removes this logarithm for all $k \ge 3$. Thus in an even stronger truncation that fixes $F_2(t)$ to its initial value $F_2^0$, analyticity can be shown uniformly in $\Lambda_0$. However, for $k=2$ the above bound is an equality, and one can verify in perturbative calculations that this bound is saturated for $k=2$, so that a more detailed analysis is required to determine the signs in the flow of $f_2(\Psi)$, to decide (depending on the sign of $\alpha$ whether there really is a pole when $\alpha \ln\frac{\Lambda_0}{m}$ becomes of order one (known as the Landau pole), or whether asymptotic freedom holds, in which case the upper bound for $\alpha$ is uniform in $\Lambda_0$. 

\section{Conclusion}\label{conclusion}

Wilsonian renormalization in the form of Polchinski's equation requires the solution of a nonlinear heat equation in a very high-dimensional space. The high dimensionality corresponds to the very large number of degrees of freedom in models of quantum field theory and statistical mechanics when a regularization is present. In the limit of interest, when the regularization is removed, the dimension of this space tends to infinity, and the equation formally becomes a functional differential equation. The solution of Polchinski's equation generates all correlation functions of the model, hence gives full information about the quantum field theory. In rigorous studies, it is necessary to give tight bounds on all correlation functions. The art of the game is to do this without needing to track all details, which would be impossible in most interesting cases.  

The majorant method introduced in \cite{BryKen,BryWri} very elegantly provides bounds for this nonlinear partial differential equation (PDE)  in terms of a PDE involving only two variables, the flow time $t$ and a variable $z$ that, roughly speaking, corresponds to the overall size of the field, thus allowing for convergence proofs that are much simplified, yet retain the power of the renormalization group method. The first purpose of our work was to revisit \cite{BryWri} and to address the gap in their method of controlling Polchinski's equation. For this reason, we followed \cite{BryWri} closely,  also taking into account some remarks from \cite{SW}. In particular, we used the same family of norms on the Grassmann algebra; the only relevant difference is that, in our work, the latter was defined as a power series of the norm parameter $z^{2}$ instead of $z$. This is of course possible because the renormalization group transformation preserves parity. Although simple, this modification is essential to develop the method. As mentioned in \cite{SW} and reiterated in Section \ref{rgtransformation} of the present paper, the gap in \cite{BryWri} has its roots in the definition of $\sigma_{(s,t)}$, which involves a square root, hence makes the function behave non-smoothly as $\sqrt{|t-s|}$ when $s \to t$. In our analysis, we have taken great care to keep the feature that only $\sigma_{(s,t)}^{2}$ appears in the bounds of the flow equation. 

Specifically, implementing these evenness restrictions entails that the terms in the majorants have a ``backward-forward'' symmetry in $z$. That is, both backward $\phi(s,\sigma_{(s,t)}-z)$ and forward $\phi(s,\sigma_{(s,t)}+z)$ translations in $z$ arise naturally in the bounds, because only the sum of the two is a function of $\sigma_{(s,t)}^2$. Of course, because the method is all about estimating norms, we could have employed further bounds on $\phi(s,\sigma_{(s,t)}-z)$ so to get only contributions for a forward translation $\phi(s,\sigma_{(s,t)}+z)$ instead, as done in \cite{BryWri}. However, this would make odd powers of $\sigma_{(s,t)}$ reappear. 

An important difference between our majorant equation and the one initially obtained in \cite{BryWri} is that in our equation, the final time also appears explicitly in the term involving the initial condition, thus in a sense we have a terminal as well as an initial problem. We deal with this by using an auxiliary equation which puts the dependence on the final time in the initial condition, but which is a standard Hamilton-Jacobi equation in the flow time $s \le t$ otherwise. Apart from the special nature of the initial condition, the Hamilton-Jacobi equation is the same as the one obtained in \cite{BryKen} for bosons. It corresponds to a nonconvex Hamiltonian, and the flow develops singularities in finite time, so one can only get short-time existence. In the case of bosons, this means that we have lost relative signs between tree graphs in the expansion of the effective action (or Mayer series, in this case). 

The second purpose of this paper was to combine the majorant with scaling ansatzes, to show how it can be used for quantum field theoretical models in spite of the short-time restriction on the solution. We have made this explicit for a fermionic $\phi^{4}$ model, and given estimates for a flow in which corrections of self-energy type are projected out. This does not provide a full construction of such a model, but it gives rigorous power counting bounds for the contributions that are, when expanded in Feynman graphs, usually called ``completely convergent graphs''. Not unexpectedly, in this field-theoretic application, the short-time restriction on the solution of the flow equation leads to a weak-coupling condition for the initial interaction, which in four dimensions displays a logarithmic behaviour on the ultraviolet cutoff. The projection can be avoided by a more careful analysis, which is the subject of ongoing work. 

\begin{appendices}

\section{Proof of Proposition \ref{Fseqlem}}\label{proofofproposition}

This section is dedicated to the demonstration of Proposition \ref{Fseqlem}. For the reader's convenience, we restate this proposition below.

\begin{prop} For each fixed $t \in [0,T]$ and $k \in \mathbb{N}$, the norm coefficients of the effective action satisfy
\begin{eqnarray}
F_{k}(t) \le \sum_{m=k}^{n}F_{m}^{(0)}\binom{2m}{2k}\sigma_{(0,t)}^{2(m-k)}+\frac{1}{2}\int_{0}^{t}ds ||\dot{A}(s)|| \sum_{\substack{l,m \\ l+m \ge k+1}} F_{l}(s)F_{m}(s)\Gamma_{l,m}(\sigma_{(s,t)}), \nonumber
\end{eqnarray}
with
\begin{eqnarray}
\Gamma_{l,m}(\xi) := 4lm \sum_{\substack{\textup{$k'$, $k''$ odd} \\k'+k'' = 2k}}\binom{2l-1}{k'}\xi^{2l-1-k'}\binom{2m-1}{k''}\xi^{2m-1-k''}. \nonumber
\end{eqnarray}
\end{prop}
\begin{proof}
Because the unnormalized effective action $\tilde{F}(t,\Psi)$ satisfies the flow equation (\ref{RGT12}), it can be expressed in terms of an integral equation
\begin{eqnarray}
\tilde{F}(t,\Psi) = (\mu_{A(t)}*\tilde{F}^{(0)})(\Psi) - \frac{1}{2}\int_{0}^{t}ds (\mu_{A_{[s,t]}}*\langle \nabla \tilde{F}(s,\cdot), \dot{A}(s)\nabla \tilde{F}(s,\cdot)\rangle)(\Psi). \label{PRO1}
\end{eqnarray}
Consequently, we have
\beq
\frac{\partial \tilde{F}(t,\Psi)}{\partial \Psi_{J}}\bigg{|}_{\Psi = 0} = \E_{\mu_{A(t)}}\bigg{[}\frac{\partial \tilde{F}^{(0)}}{\partial \Psi_{J}}\bigg{]} -\frac{1}{2}\int_{0}^{t}ds \sum_{K\subset J}(-1)^{|K|}\E_{\mu_{A_{[s,t]}}}\bigg{[}\bigg{\langle} \nabla \bigg{(}\frac{\partial \tilde{F}(s,\cdot)}{\partial \Psi_{K}}\bigg{)},\dot{A}(s)\nabla \bigg{(}\frac{\partial \tilde{F}(s,\cdot)}{\partial \Psi_{J\setminus K}}\bigg{)}\bigg{\rangle}\bigg{]}, \label{PRO2}
\eeq
by Leibniz rule. Since $F$ an $\tilde{F}$ only differ by a normalization constant, it follows that $F$ must satisfy the same equation, that is
\beq
\frac{\partial F(t,\Psi)}{\partial \Psi_{J}}\bigg{|}_{\Psi = 0} = \E_{\mu_{A(t)}}\bigg{[}\frac{\partial F^{(0)}}{\partial \Psi_{J}}\bigg{]} -\frac{1}{2}\int_{0}^{t}ds \sum_{K\subset J}(-1)^{|K|}\E_{\mu_{A_{[s,t]}}}\bigg{[}\bigg{\langle} \nabla \bigg{(}\frac{\partial F(s,\cdot)}{\partial \Psi_{K}}\bigg{)},\dot{A}(s)\nabla \bigg{(}\frac{\partial F(s,\cdot)}{\partial \Psi_{J\setminus K}}\bigg{)}\bigg{\rangle}\bigg{]}. \label{PRO3}
\eeq
Notice that we are only interested in the case $|J| = 2k$ has even cardinality. Getting back to the notations introduced in (\ref{NOR6}) and expressing $F^{(0)}$ by
\begin{eqnarray}
F^{(0)}(\Psi) = \sum_{J\subset \{1,...,2n\}}\zeta^{(0)}_{J}\Psi_{J}, \label{PRO4}
\end{eqnarray}
we can readily estimate the first term in the right hand side of (\ref{PRO3})
\begin{eqnarray}
\bigg{|} \E_{\mu_{A(t)}}\bigg{[}\frac{\partial F^{(0)}}{\partial \Psi_{J}}\bigg{]}\bigg{|} \le |\sum_{K\supset J}\zeta_{K}^{(0)}\E_{\mu_{A(t)}}[\Psi_{K\setminus J}]| \le \sum_{m=k}^{n}\sum_{\substack{K\supset J \\ |K| = 2m}}|\zeta_{K}^{(0)}|\sigma_{(0,t)}^{2(m-k)}, \label{PRO5}
\end{eqnarray}
in view of (\ref{RGT8}). Hence,
\beq
\begin{split}
\sup_{i_{0}}\frac{1}{2k}\sum_{J \ni i_{0}, |J| = 2k}\bigg{|} \E_{\mu_{A(t)}}\bigg{[}\frac{\partial F^{(0)}}{\partial \Psi_{J}}\bigg{]}\bigg{|} &\le \sum_{m=k}^{n}\sup_{i_{0}}\frac{1}{2m}\sum_{K \ni i_{0}, |K| = 2m}|\zeta_{K}^{(0)}|\frac{2m}{2k}\sum_{\substack{J \subset K, |J| = 2k \\ i_{0} \in J}}\sigma_{(0,t)}^{2(m-k)} \\
&= \sum_{m=k}^{n}F_{m}^{(0)}\binom{2m}{2k}\sigma_{(0,t)}^{2(m-k)}. 
\end{split}
\label{PRO6}
\eeq
As for the second term in the right hand side of (\ref{PRO3}), we first set
\begin{eqnarray}
F(t,\Psi) = \sum_{J\subset \{1,...,2n\}}\zeta_{J}\Psi_{J}, \label{PRO7}
\end{eqnarray}
in such a way that
\beq
\begin{split}
&\bigg{|}\E_{\mu_{A_{[s,t]}}}\bigg{[}\bigg{\langle} \nabla \bigg{(}\frac{\partial F(s,\cdot)}{\partial \Psi_{K}}\bigg{)},\dot{A}(s)\nabla \bigg{(}\frac{\partial F(s,\cdot)}{\partial \Psi_{J\setminus K}}\bigg{)}\bigg{\rangle}\bigg{]} \bigg{|} 
\\
&\le \sum_{i,j}|\dot{a}_{ij}|\sum_{L \supset K\cup \{i\}}|\zeta_{L}|\sum_{M\supset (J\setminus K)\cup \{j\}}|\zeta_{M}|\E_{\mu_{A_{[s,t]}}}[ \Psi_{L\setminus (K\cup \{i\})}\wedge \Psi_{M\setminus ((J\setminus K)\cup \{j\})}] \\
& \le\sum_{i,j}|\dot{a}_{ij}|\sum_{L\supset K\cup \{i\}}|\zeta_{L}|\sum_{M\supset (J\setminus K)\cup\{j\}}|\zeta_{M}|\sigma_{(s,t)}^{|L|+|M|-|K|-|J\setminus K|-2}\chi_{L\cap M = \emptyset} 
\end{split}
\eeq
where $\chi_{L\cap M=\emptyset}$ indicates that the sum ranges over disjoint sets $L$ and $M$. Summing over $K\subset J$, with $|J| = 2k$, the former is itself majorized by
\beq
\begin{split}
&\sum_{i,j}|\dot{a}_{ij}|\sup_{i_{0}}\frac{1}{2k}\sum_{\substack{J \ni i_{0}\\|J|=2k}}\sum_{K\subset J}\sum_{L\supset K\cup \{i\}}|\zeta_{L}|\sum_{M\supset (J\setminus K)\cup \{j\}}|\zeta_{M}|\sigma_{(s,t)}^{|L|+|M|-|K|-|J\setminus K|-2}\nonumber \\
&\le \sum_{i,j}|\dot{a}_{ij}|\sup_{i_{0}}\frac{1}{2k}\sum_{\substack{J\ni i_{0}\\|J|=2k}}\sum_{K\subset J}\sum_{\substack{l,m\\ l+m\ge k+1}}\sum_{\substack{L\supset K\cup \{i\}\\|L|=2l}}|\zeta_{L}|\sum_{\substack{M\supset (J\setminus K)\cup\{j\}\\|M|=2m}}|\zeta_{M}|\sigma_{(s,t)}^{2l+2m-|K|-|J\setminus K|-2}. \label{PRO8}
\end{split}
\eeq
In (\ref{PRO8}), $i_{0}\in J$ could be an element of either $K$ or $J\setminus K$. If $i_{0} \in K$, (\ref{PRO8}) is further majorized by
\beq
\sup_{\tilde{j}}\sum_{i=1}^{2n}|\dot{a}_{i\tilde{j}}(s)| \sum_{\substack{l,m \\ l+m \ge k+1}} 4lm\; \sigma_{(s,t)}^{2l+2m-2k-2}\sup_{i_{0}}\frac{1}{2l}\sum_{\substack{L\ni i_{0}\\|L|=2l}}|\zeta_{L}|\sup_{j_{0}}\frac{1}{2m}\sum_{\substack{M\ni j_{0}\\|M|=2m}}|\zeta_{M}|\sum_{j\in L}\frac{1}{2k}\sum_{\substack{K\subset L\setminus \{i\}\\K'\subset M \setminus\{j\}\\ K\cup K'=J\ni i_{0}\\ |J|=2k}}1 \label{PRO9}
\eeq
If, on the other hand, $i_{0} \in J\setminus K$ then instead of the above majorant, (\ref{PRO8}) is majorized by
\beq
\sup_{\tilde{i}}\sum_{j=1}^{2n}|\dot{a}_{\tilde{i}j}(s)| \sum_{\substack{l,m \\ l+m \ge k+1}}4lm\;\sigma_{(s,t)}^{2l+2m-2k-2}\sup_{j}\frac{1}{2l}\sum_{\substack{L\ni j\\|L|=2l}}|\zeta_{L}|\sup_{j_{0}}\frac{1}{2m}\sum_{\substack{M\ni j_{0}\\|M|=2m}}|\zeta_{M}|\sum_{i\in L}\frac{1}{2k}\sum_{\substack{K\subset L\setminus \{j\}\\K'\subset M \setminus\{i\}\\ K\cup K'=J\ni i_{0}\\ |J|=2k}}1. \label{PRO10}
\eeq
In either case, both expressions (\ref{PRO9}) and (\ref{PRO10}) are bounded by 
\begin{eqnarray}
||\dot{A}||\sum_{\substack{l,m \\ l+m\ge k+1}} F_{l}(s)F_{m}(s)\Gamma_{l,m}(\sigma_{(s,t)}). \nonumber
\end{eqnarray}
Indeed, when $i_{0} \in K$ we get
\begin{align}
\sum_{\substack{\textup{$k'$,$k''$ odd}\\ k+k'=2k}}\sum_{j\in L}\frac{1}{2(k'+k'')}\sum_{\substack{K\subset L\setminus\{i\}\\|K|=k'\\ K\ni i_{0}}}\sum_{\substack{K'\subset M\\|K'|=k''+1\\K'\ni j}}1 &= \sum_{\substack{\textup{$k'$,$k''$odd}\\k'+k''=2k}}\frac{k''+1}{k'+k''}\binom{2l-1}{k'}\binom{2m-1}{k''} \nonumber \\
&\le \sum_{\substack{\textup{$k'$,$k''$ odd}\\k'+k''=2k}}\binom{2l-1}{k'}\binom{2m-1}{k''} \nonumber
\end{align}
and a similar argument proves this inequality holds when $i_{0} \in J\setminus K$.  
\end{proof}


\newcommand{\vphiz}{\vphi_0} 
\section{Hamilton-Jacobi theory and the proof of Theorem \ref{newmajorant}}\label{new-AppendixB}

Here we study (\ref{HamJacaux}) and prove existence, uniqueness, and analyticity, of its solution. 
The $t$-dependence in the function $\psi$ that is to solve (\ref{HamJacaux}) only arises from the initial condition. We now consider the initial value problem 
\beq
 \begin{split}\label{HamJacaux-bis}
 \frac{\del \vphi}{\del s} (s,z)
& =
 \frac12 
 \Vert\dot{A}(s)\Vert
\left(
\frac{\partial \vphi(s,z)}{\partial z}
\right)^2   
\\
 \vphi(0,z) 
 &= 
\vphiz (z) \; ,
 \end{split}
 \eeq
assuming that $\vphiz$ is even and analytic in some disk $|z| < R$ around zero, and has nonnegative coefficients when expanded in powers of $z$. The $t$-dependent initial condition of (\ref{HamJacaux}) can be accommodated by choice of $\vphiz$. We first change variables 
\beq\label{varchange}
t \to \tau (t) = \int_0^t \Vert \dot A (s)\Vert \; \dd s\; .
\eeq
Assuming that $\Vert \dot A (s)\Vert  > 0$ for all $s$, this is a regular change of coordinates, and it transforms 
(\ref{HamJacaux-bis}) to 
 \beq\label{HamJac0}
\frac{\del \vphi}{\del s} (s,z)
=
\frac12 \left( \frac{\del \vphi}{\del z } (s,z)\right)^2
\qquad
\vphi(0,z)
=
\vphiz(z) \; .
\eeq
(where we have replaced $\tau(s)$ again by $s$ in the notation, to avoid a proliferation of new notations in this section). 
We briefly recall the method of characteristics which is used to solve this equation:
If (\ref{HamJac0}) has a solution $\vphi$ that is $C^2$ in $z$, then $u(s,z) =\frac{\del \vphi}{\del z} (s,z) $ solves
\beq\label{HamJac1}
\dot u - u u' = 0 \; , \qquad 
u_0 (z)  = \vphiz'(z)\;.
\eeq
where we have denoted the $s$-derivative by a dot and the $z$ derivative by a prime.  
Let $w$ be in the domain of analyticity of $u_0$. Assume that there is a $C^1$-function $s \mapsto z(s)$ with $z(0)=w$ that solves
\beq
\frac{\dd z}{\dd s} (s) = - u(s, z(s))
\eeq
and set $U(s) = u(s,z(s))$. Then 
\beq
\frac{\dd U}{\dd s} =
\dot u + \dot z\; u'
=
\dot u - u \; u' 
=
0
\eeq
by (\ref{HamJac1}). Thus $U$ is constant: $U(s) = U(0) = u_0(w)$, and $\dot z (s) = - u_0 (w)$, hence
\beq\label{HamJac5}
z(s)
=
w - u_0(w) \; s \; .
\eeq 
The idea is that if (\ref{HamJac5}) can be solved for $w$ as a function of $s$ and $z$, this gives a solution of (\ref{HamJac1}). This is possible for small enough $|s|$ because $u_0$ is analytic in $z$, $z(0)=w$, and 
\beq\label{HamJac5d}
\frac{\del z}{\del w}
=
1 - s\; u'_0(w) \; ,
\eeq
so $\frac{\del z}{\del w} \ne 0$ if $s < \frac{1}{|u_0'(w)|}$. By (\ref{HamJac5}), $z$ is also analytic (since affine) in $s$. 

\begin{lemma}\label{lemb1}
For $z_0 \in \bC$ and $r > 0$ denote $B_r (z_0) = \{ z \in \bC: |z-z_0| < r\}$. 

\begin{enumerate}

\item
There are $s_0>0$ and $r_0>0$ and a unique function $\omega: (-s_0,s_0) \times B_{r_0} (0) \to \bC$, $(s,z) \mapsto \omega (s,z) $ such that $w = \omega(s,z)$ solves  $(\ref{HamJac5})$, i.e.
\beq\label{HamJac6}
z
=
\omega (s,z) - s\; u_0(\omega (s,z)) \; ,
\eeq
and $\omega$ is $C^1$ in $s$ and analytic in $z$ for $|z| < r_0$. 
\item
The function 
\beq\label{usol}
u(s,z)
=
u_0 (\omega (s,z))
\eeq
then solves $(\ref{HamJac1})$. 

\end{enumerate}

\end{lemma}

\begin{proof}
(1) The function $\omega$ exists by the inversion theorem for analytic functions in a sufficiently small interval $(-s_0,s_0)$ and it is $C^1 $ (in fact, analytic) there.

(2) Differentiation of (\ref{HamJac6}) with respect to $s$ and $z$ gives
\beq
\begin{split}
0
&=
\dot \omega - u_0 (\omega) -s u_0'(\omega) \; \dot \omega
=
- u_0 (\omega) + \dot \omega ( 1 -s u_0'(\omega) ) 
\\
1
&=
\omega' - s u_0' (\omega) \; \omega'
=
\omega'  (1 - s u_0' (\omega)) \; .
\end{split}
\end{equation}
so $\dot \omega = u_0(\omega) \; \omega'  = u \omega'$ by (\ref{usol}), 
which implies (\ref{HamJac1}).
\end{proof}

More detailed estimates for the radius of convergence in $z$ and for analyticity in $s$ are a nice application of B\" urmann-Lagrange inversion, which goes as follows. 

\begin{lemma}
Let $R > 0$ and $\vphiz$ be analytic in $z$ for $|z| < R$. Furthermore, let $\vphiz$ be even in $z$ and $\vphiz(0) = 0$.
Then there are $s_1 > 0$ and $r_1 > 0$ such that the initial-value problem $(\ref{HamJac0})$ has a unique solution 
\beq
\vphi: B_{s_1}(0) \times B_{r_1} (0) \to \bC\;,
\qquad
(s,z) \mapsto \vphi (s,z) \;,
\eeq
and $\vphi$ is analytic both in $s$ and in $z$ on its domain of definition. In particular, for all $n \in \bN_0$ the expansion coefficients
$\vphi_n (s)$ in 
\beq
\vphi (s,z) 
=
\sum_{n=0}^\infty 
\vphi_n (s) \; z^n
\eeq
are analytic functions of $s$ on $\{ s \in \bC: |s| < s_1\}$. 

If $\vphi_n (0) =0 $ for odd $n$, then $\vphi_n (s) = 0 $ for odd $n$ and all $ s \in B_{s_1} (0) $.

If $\vphi_n (0) \ge 0$ for all $n \in \bN$ , then $\vphi_n (s)  \ge 0$ for all $s \in [0,s_1)$ and all $n \in \bN$. 

\end{lemma}

\begin{proof}
The function $u_0 = \vphiz'$ is analytic on $B_R(0)$. Call 
\beq
f_s (w) = w - u_0 (w) \; s 
\eeq
We fix some $z_0$ and search for a solution curve $w \mapsto f_s (w)$ that contains $z_0$ for all sufficiently small $s$, i.e.\ some $w_0 (s)$ for which $f_s (w_0) = z_0$. Choose $z_0 = 0$. 
Because by hypothesis $\vphiz$ is even in $z$, $u_0 (0) = 0$, so $f_s (0) =0$ for all $s$, so $w_0 = 0$ satisfies $f_s (w_0) = 0$ for all $s$.\footnote{If this assumption is not made, then $w_0$, satisfying $f_s (w_0) = z_0$ will depend on $s$ because we have fixed $z_0$ independent of $s$. The condition then reads $ z_0 = 0  = w_0 - s u_0 (w_0) = u_{0,0} + w_0 (1- s u_{0,1}) + O(w_0^2)$.which can always be solved for $w_0$ if $s$ is small enough, even if $u_{0,0}=0$.}
As stated in Lemma \ref{lemb1}, because $\frac{\dd f_s}{\dd w} (0) = 1 - s \; u_0'(0) \ne 0$ for $|s| < \frac{1}{|u_0'(0)|}$, the inversion theorem for analytic functions applies, and there is an $s'$, $0< s' <  \frac{1}{|u_0'(0)|}$ and $r'>0$ such that for all $|s| < s'$,
the inverse function $g_s$ satisfying $g_s (f_s (w))=w$ exists and is unique and analytic on $\{ z \in \bC: |z|<r'\}$.  For all these $z$, 
\beq
g_s (z) 
=
\sum_{n=0}^\infty g_n (s) \; z^n\;, 
\eeq
By construction, $g_0(s) = w_0 = 0$. For $n \ge 1$, we use the B\"urmann-Lagrange formula
\beq
g_n (s) 
=
\frac{1}{2\pi \I n}
\int_{|w|=\delta} \frac{\dd w}{f_s(w)^n}
=
\frac{1}{2\pi \I n}
\int_{|w|=\delta} \frac{\dd w}{(w- s u_0(w))^n}
\eeq
where $\delta > 0$ is so small that $\inf_\theta |f_s (\delta e^{\I \theta})| > 0$. Since 
\beq
w- s u_0(w)
=
w (1 - s u_0'(0)) + s w^2 \; R_1 (w)
\eeq
where the Taylor remainder $R_1$ is analytic in $w$, 
\beq
g_n (s) 
=
\frac{1}{(1 - s u_0'(0))^n} \;
\frac{1}{2\pi \I n} 
\int_{|w|=\delta} 
\frac{\dd w}{(w^n (1 + \rho(s,w))^n}
\eeq
with 
\beq\label{rhodef}
\rho(s,w) = \frac{sw R_1(w)}{1 - s u_0'(0)}\; .
\eeq
Because there is $0 < s'' < s'$ such that $|\rho(s,w)| < 1$ for all $|s| < s''$ and all $|w|=\delta$, it is now explicit by geometric series expansions that for all $n \in \bN$, the coefficients $g_n$ are analytic functions of $s$ on $B_{s''}(0)$. By compact convergence of the power series in $z$, this implies that $s \mapsto g_s (z)$ is analytic. 

We now choose $s_1 > 0$ so small that $s_1 < s''$. 

Because convergent power series can be integrated termwise, 
\beq
\vphi (s,z) 
=
\sum_{n=0}^\infty \vphi_n (s) \; z^n
\eeq
with $\vphi_n (s) = \frac{1}{n} g_{n-1} (s)$. Thus also all $\vphi_n$ are analytic on the same disk $\{ s \in \bC : |s| < s''\}$, and the statement about analyticity of $\vphi$ is proven. 

The differential equation preserves evenness of $\vphi$ in $z$, so the statement of the lemma about the coefficients for odd $n$ is obvious. 

It remains to prove that the coefficients are nonnegative. 
The differential equation and a comparison of coefficients of powers of $z^n$ imply that 
\beq
\frac{\dd \vphi_m}{\dd s} (s) 
=
\sum_{\ell=1}^m \ell \; (m-\ell+1)\;\vphi_\ell (s) \; \vphi_{m-\ell+1} (s) \; .
\eeq 
Similarly, the derivative of order $n$ is a linear combination of lower derivatives with nonnegative coefficients. Explicitly, by Leibniz's rule
\beq\label{leibspeis}
\frac{\dd^n \vphi_m}{\dd s^n} (s) 
=
\sum_{\ell=1}^m \ell \; (m-\ell+1)\;
\sum_{\nu =1}^{n-1}
{n-1 \choose \nu}\;
\frac{\dd^{n-1-\nu} \vphi_\ell}{\dd s^{n-1-\nu}} (s) \; 
\frac{\dd^\nu \vphi_{m-\ell+1}}{\dd s^\nu} (s) \; .
\eeq
Let $I=[0,s_1]$. We now show that the set $P$ of those $s$ for which $\frac{\dd^n \vphi_m}{\dd s^n} (s) \ge 0$ for all $m$ and $n$ is nonempty, and both open and closed as a subset of $I$. Because $I$ is connected, $P=I$ must then hold, which implies the statement of the lemma about positivity. Let
\beq
P_{m,n} = \{ s \in [0,s_1]: \sfrac{\dd^n \vphi_m}{\dd s^n} (\sigma) \ge 0 \mbox{ for all } \sigma \in [0,s]\}
\eeq
and $P = \bigcap_{m,n \ge 0} P_{m,n}$.
By hypothesis on the initial condition, $0 \in P$. 
Because every $\sfrac{\dd^n \vphi_m}{\dd s^n}$ is continuous in $s$, $P_{m,n}$ is a closed subset of $I$, hence $P$ is closed in $I$. Now $P \subset I$, and $I$ is a compact subset of $ \{ s \in \bC : |s| < s''\}$, so it has positive distance to the boundary of $ \{ s \in \bC : |s| < s''\}$. Thus there is $\veps > 0$ such that for all $s \in P$, all $|h| <\veps$, and all $m \in \bN$, the expansion
\beq\label{sesum}
\vphi_m (s+h)
=
\sum_{n=0}^\infty 
\frac{1}{n!}\;
\frac{\dd^n \vphi_m}{\dd s^n} (s) \; h^n
\eeq 
converges. (Note that $\veps$ does not depend on $n$.) Because $s \in P$, and by (\ref{leibspeis}), $\frac{\dd^n \vphi_m}{\dd s^n} (s) \ge 0$. Thus the sum in (\ref{sesum}) and all its $s$-derivatives are positive for $h > 0$, so  $s+h \in P$ for all $h \in [0,\veps)$. Moreover, $[s,s-\veps) \cap I \subset P$ by definition of $P$. Thus $(s-\veps,s+\veps) \cap I \subset P$, hence $P$ is open in $I$. 
\end{proof}

We now return to the $t$-dependent initial condition. 
We can choose $s_1$ so small that for all $t \in [0,s_1]$, $\sigma_{(0,t)}$ is less than $R_0/2$, where $R_0$ is the radius of convergence of the initial condition $\Fphz$. Then $\vphiz$ has analyticity radius $R \ge R_0/2$, and the dependence on $\sigma_{s,t}^2$ is analytic, too. Choose $r_1 \le \min \{ r',R\}$.

Thus by choosing $t_0 = s_1$ sufficiently small, we have obtained a unique solution to (\ref{majaux}). By setting $s=t$, we then obtain a solution of  (\ref{majeq}). Analyticity implies that the sequence of $\phi_k (s)$ that solves (\ref{phi-eq}) exists and is unique. Thus Theorem \ref{newmajorant} holds.

\end{appendices}

\bigskip
\noindent
{\bf Acknowledgement. } The authors would like to thank David Brydges for his insightful comments and suggestions on the content of this paper, and the referee for pointing out \cite{Greenblatt}. WK would like to thank Conselho Nacional de Desenvolvimento Cient\'ifico e Tecnol\' ogico (CNPq) for the financial support during his PhD. This study was financed in part by the Coordena\c c\~ ao de Aperfei\c coamento de Pessoal de N\' ivel Superior (CAPES) - Finance Code 001. This work is supported by Deutsche Forschungsgemeinschaft (DFG, German Research Foundation) under Germany's Excellence Strategy  EXC-2181/1 - 390900948 (the Heidelberg STRUCTURES Cluster of Excellence).

The authors declare no competing interests. No datasets were generated or analysed during the current study. 


\end{document}